\pdfoutput=1 
\documentclass[mfin,12pt]{UNSWthesis}    

\usepackage{amsfonts}
\usepackage{amssymb}
\usepackage{amsthm}
\usepackage{latexsym,amsmath}
\usepackage{graphicx}
\usepackage{afterpage}
\usepackage{mathrsfs}
\usepackage{enumitem}

\newcommand{\R}{\mathbb{R}}
\newcommand{\Q}{\mathbb{Q}}
\newcommand{\C}{\mathbb{C}}
\newcommand{\NN}{\mathbb{N}}
\newcommand{\N}{\mathscr{N}}
\newcommand{\F}{\mathscr{F}}
\newcommand{\G}{\mathscr{G}}
\newcommand{\FF}{\mathbb{F}}
\renewcommand{\P}{\mathbb{P}}

\newcommand{\B}{\mathscr{B}}

\newcommand{\TT}{\mathcal{T}}
\newcommand{\LL}{\mathcal{L}}

\newcommand{\J}{\mathcal{J}}
\newcommand{\E}{\mathbb{E}}

\newcommand{\A}{\mathcal{A}}
\newcommand{\ZZ}{\mathcal{Z}}
\newcommand{\h}{\text{H\"{o}lder }}
\newcommand{\bd}{\text{B\"auerle and Desmettre }}

\def\tr(#1){{\rm trace}(#1)}
\def\Exp(#1){{\mathbb E}(#1)}
\def\Exps(#1){{\mathbb E}\sparen(#1)}

\newcommand{\de}{\overset{\Delta}{=}}
\newcommand{\sep}{\, | \,}
\newcommand{\fa}{\, \forall \;}

\def\implies{\Rightarrow}
\def\inpr#1,#2{\t \hbox{\langle #1 , #2 \rangle} \t}
\def\ip<#1,#2>{\langle #1,#2 \rangle}

\def\paren(#1){\left( #1 \right)}

\def\sparen(#1){\Bigl ( #1 \Bigr )}

\newtheorem{theorem}{Theorem}[section]
\newtheorem{assumption}[theorem]{Assumption}

\newtheorem{proposition}[theorem]{Proposition}

\newtheorem{definition}[theorem]{Definition}

\newtheorem{remark}[theorem]{Remark}

\numberwithin{equation}{section}

\title{PORTFOLIO OPTIMISATION: \vskip 0.25cm
UNDER ROUGH HESTON MODELS}

\authornameonly{Benjamin James Duthie}

\author{\Authornameonly\\{\bigskip}Supervisor: Associate Professor Spiridon Penev}

\submitdate{9th of August 2019}

\copyrightfalse
\figurespagefalse
\tablespagefalse

\begin{document}

\beforepreface



\prefacesection{Abstract}

This thesis investigates Merton's portfolio problem under two different rough Heston models, which have a non-Markovian structure. The motivation behind this choice of problem is due to the recent discovery and success of rough volatility processes. The optimisation problem is solved from two different approaches: firstly by considering an auxiliary random process, which solves the optimisation problem with the martingale optimality principle, and secondly, by a finite dimensional approximation of the volatility process which casts the problem into its classical stochastic control framework. In addition, we show how classical results from Merton's portfolio optimisation problem can be used to help motivate the construction of the solution in both cases. The optimal strategy under both approaches is then derived in a semi-closed form, and comparisons between the results made. The approaches discussed in this thesis, combined with the historical works on the distortion transformation, provide a strong foundation to build models capable of handling increasing complexity demanded by the ever growing financial market.

\afterpreface

%
%


\chapter{Introduction}\label{s-intro}

There has recently been growing interest in the study of rough volatility models (see \cite{diehl_et_al_2017_stochastic_control_rough_paths}, \cite{euch_rosenbaum_2019_characteric_function_rough_heston_models}, \cite{forde_zhang_2017_asymnptotics_rough_volatility}, \cite{fukasawa_et_al_2019_is_volatility_rough}, \cite{garnier_solna_2018_option_pricing_ou_rough}, \cite{guennoun_et_al_2018_asymptotic_behavior_fractional_heston} and \cite{Keller-Ressel_et_al_2018_affine_rough_models}) since the seminal paper by Gatheral et al.  (see \cite{gatheral_et_al_2018_volatility_is_rough}) titled ``Volatility is Rough'', which first appeared on the SSRN in 2014. The discovery that volatility exhibited rough behaviour was established from a regression analysis on high frequency futures contract data whereby Gatheral et al. deduced the following relationship
\begin{equation}\label{g-scaling}
\E[|\log(\sigma_\Delta)-\log(\sigma_0)|^q] \approx  \frac{1}{n}\sum_{t=1}^{n}|\log(\hat{\sigma}_{t+\Delta})-\log(\hat{\sigma}_{t})|^q \approx K_q \Delta^{\zeta_q}, 
\end{equation}
for differing values of $q$.
This relationship was achieved by using arguments of stationarity and the following result from Mandelbrot and Van Ness (see \cite{mandelbrot_ness_1968_fractional_bm})
\begin{equation}
 \E[|W_{t+\Delta}^H-W_t^H|^q]=K_q\Delta^{qH}, \quad 0\le t < \infty, \, \Delta \ge 0, \, q>0,
\end{equation}
 where $(W_t^H)_{0\le t < \infty}$ is a fractional Brownian motion (fBm) with Hurst parameter $H\in (0,1)$, and $K_q$ is the moment of order $q$ of the absolute value of a standard Gaussian variable. By then performing another regression analysis of $\zeta_q$ against $q$, they deduced the relationship $\zeta_q \approx Hq$, with $H\approx 0.1$. 

The parameter $H$ is referred to as the Hurst parameter which was originally introduced by Mandelbrot and Van Ness in 1968 (see \cite{mandelbrot_ness_1968_fractional_bm}). The Hurst parameter has the following relationship with the fBm which we recall for the readers convenience.


\begin{definition}[Fractional Brownian motion (fBm)]\cite{mandelbrot_ness_1968_fractional_bm} Let $H\in (0,1)$, $b_0 \in \R$, and $W$ a Brownian motion. Then the fBm $W^H$ with Hurst parameter $H$ is defined as the Weyl fractional integral of $W$ with $\lbrace W^H_t \sep t>0 \rbrace $ given by 
\begin{equation}
\begin{split}
W^H_t = W^H_0 + \frac{1}{\Gamma(H+1/2)}\bigg\lbrace \int_{-\infty}^{0}[(t-u)^{H-1/2}-(-u)^{H-1/2}]&dW_u \\
+ \int_{0}^{t}(t-u)^{H-1/2}&dW_u \bigg\rbrace.
\end{split}
\end{equation}
Then $W^H$ is a centred Gaussian process with $W^H_0 = b_0$ and covariance function
\begin{equation}\label{fb-cov}
\E[W_t^HW_s^H] = \frac{1}{2}\left(t^{2H}+s^{2H}-|t-s|^{2H}\right).
\end{equation}
\end{definition}
\noindent From (\ref{fb-cov}) we obtain the following relationships between $H$ and $W^H$. Let $0<s_1 < s_2 <t_1<t_2$, then 
$$\E [(W_{s_1}^H-W_{s_2}^H)(W_{t_1}^H-W_{t_2}^H)] < 0 \quad \text{for $H \in (0,1/2)$},$$
$$\E [(W_{s_1}^H-W_{s_2}^H)(W_{t_1}^H-W_{t_2}^H)] > 0 \quad \text{for $H \in (1/2,1)$}$$
(see \cite{shevchenko_2015_fbm_nutshell} for further details). Thus, for $H \in (0,1/2)$ the fBm has the property of negatively correlated non-overlapping increments; in other words the fBm has the property of counter-persistence and therefore fluctuates violently. In contrast, for $H \in (1/2,1)$ the fBm has the property of positively correlated non-overlapping increments, and thus has the property of persistence. In this instance the fBm has the property of long-memory (long-range dependence). Therefore, the Hurst parameter ($H$) can be seen as a measure of the smoothness characteristic in the underlying volatility process. 

Gatheral et al. discovery of the scaling relationship (\ref{g-scaling}) with the stylised fact that the increments of log-volatility is close to Gaussian, which is shown in \cite{andersen_et_al_2001a} and \cite{andersen_et_al_2001b}; then suggested the following stationary fractional Ornstein--Uhlenbeck process to model the volatility 
\begin{equation}
\log \sigma_t = \nu \int_{-\infty}^{t}e^{-\alpha(t-u)}dW_u^H +m,
\end{equation}
with $m \in \R$ the long term mean level, $\alpha >0$ the mean reversion, and $\nu>0$ the instantaneous volatility. From this model they then deduced spurious long memory of volatility by means of simulation. Additionally, the resulting Rough Fractional Stochastic Volatility (RFSV) model turned out to be formally almost identical to the Fractional Stochastic Volatility (FSV) model developed by Comte and Renault in 2005 (see \cite{comte_renault_1998_long_memory_stochastic_volatility_models}), with one major difference: in the FSV model $H>1/2$ to ensure long memory, whereas in the RFSV model $H<1/2$, with typical values of $H$ close to 0.1. The value of $H$ directly affects the mean-reversion parameter $\alpha$ and consequently in the FSV model $\alpha \gg 1/T$ to ensure a decreasing term structure of the at-the-money (ATM) volatility skew for longer expiration's. Whereas, the RFSV model has a mean-reversion parameter of $\alpha \ll 1/T$. The ATM volatility skew (intrinsic value $k=0$) at maturity $\tau$ is given by

\begin{equation}\label{e-volatility-skew}
\psi (\tau) \de \left\lvert \frac{\partial \sigma_{\text{BS}}(k,\tau)}{\partial k}\right\rvert_{k=0}
\end{equation}

\noindent and has been well approximated by Fukasawa in 2011 for small $\tau$ (see \cite{fukasawa_2011_asymptotics_volatility}) by the following relationship
\begin{equation}\label{ef_atm}
\psi (\tau) \sim \tau^{H-1/2}, \text{ when $\tau \downarrow 0$.}
\end{equation}

\noindent The importance of this result was that it provided a counterexample to the widespread belief that the explosion of the volatility smile implied the presence of jumps at small $\tau$ (see \cite{carr_wu_2003_volatility_smile_jumps}). Then from (\ref{ef_atm}), if $H>1/2$ the volatility skew function is increasing in time for small values of $\tau$, which is completely inconsistent with the approximately observed volatility skew term structure of $1/\sqrt{\tau}$ (see \cite{gatheral_et_al_2018_volatility_is_rough}). Consequently, for very short expiration's ($\tau \ll 1/\alpha$), FSV models ($H>1/2$) still generate a term structure of volatility skew that is inconsistent with the observed one. Whereas, RFSV models ($H<1/2$) reproduce both the observed smoothness of the volatility process and the term structure of volatility skew. The following figure from \cite{gatheral_et_al_2018_volatility_is_rough} is of the S\&P ATM volatility skews given by (\ref{e-volatility-skew}).

\begin{figure}[h]
	\centering
	\includegraphics[width=0.7\linewidth]{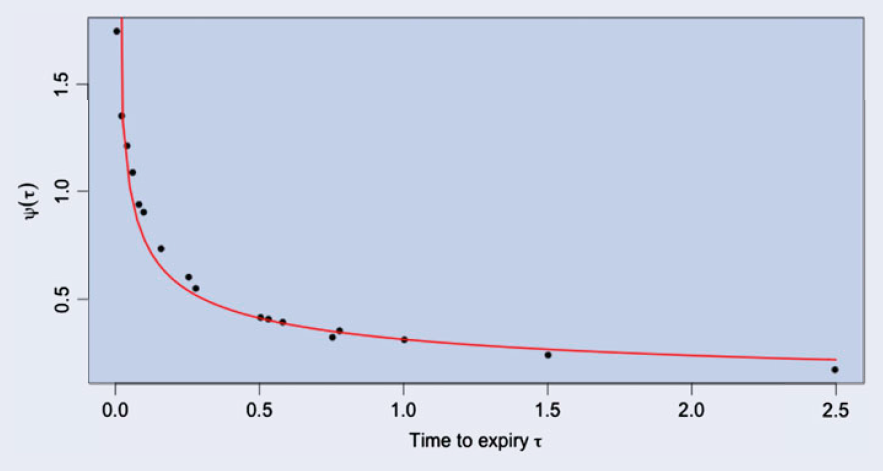}
	\caption{The black dots are non-parametric estimates of the S\&P ATM volatility skews as of 20 June 2013; the red curve is the power law fit$ ψ(\tau) = A \tau^{−0.4}$ (figure 2, \cite{gatheral_et_al_2018_volatility_is_rough}).}
	\label{fig:atmvolatilityskewexpiry}
\end{figure}

\noindent Moreover, recent studies by Fukasawa et. al. (see \cite{fukasawa_et_al_2019_is_volatility_rough}) also confirmed the accuracy of these empirical findings, thus supporting the use of RFSV models in pricing financial derivatives.


\section{Motivation and major contributions}\label{s-inspiration}
The popularity of the Heston model in the financial market setting has lead to the introduction of numerous versions of the Rough Heston model (see \cite{bauerle_desmettre_2019_portfolio_optimization_des_baulere}, \cite{euch_rosenbaum_2018_perfect_hedging_rough_heston_models}, \cite{euch_rosenbaum_2019_characteric_function_rough_heston_models}, \cite{guennoun_et_al_2018_asymptotic_behavior_fractional_heston} and \cite{jaber_et_al_2017_affine_volterra_processes}). Specifically, the work from Jaber et al. (see \cite{jaber_et_al_2017_affine_volterra_processes}) in affine Volterra processes, resulting in the development of the Volterra Heston model and thus, provides a number of useful results. For example, the rough Heston model developed in \cite{euch_rosenbaum_2019_characteric_function_rough_heston_models}, becomes a special case of the Volterra Heston model under the kernel $K(t) = \frac{t^{H-1/2}}{\Gamma(H+1/2)}$, and by the use of Riccati-Volterra equations (as shown in \cite{jaber_et_al_2017_affine_volterra_processes}), the work produced in \cite{euch_rosenbaum_2019_characteric_function_rough_heston_models} on the characteristic functions of the rough Heston model can be extended to the Volterra Heston model. In addition to these types of rough Heston models, Guennoun et al. (see \cite{guennoun_et_al_2018_asymptotic_behavior_fractional_heston}) developed an alternative definition of the rough Heston model through the use of fractional derivatives, which is extended in \cite{bauerle_desmettre_2019_portfolio_optimization_des_baulere}. The associated properties of these rough Heston models have subsequently lead to the development of two contrasting approaches to the Merton portfolio optimisation problem, which is the focus of this thesis. 

Motivated from the results in \cite{jaber_et_al_2017_affine_volterra_processes} and the development of the martingale distortion transformation (see \cite{zariphopoulou_2001_supplement_martingale_distortion_markovian}, \cite{foque_hu_2018a_optimal_portfolio_fast_smooth}, \cite{foque_hu_2018b_optimal_portfolio_fast_rough}, \cite{tehranchi_2004_martingale_distortion_non_markovian}), Han and Wong (see \cite{han_wong_portfolio_utility_volterra_heston}) propose an approach to the optimisation problem by means of the martingale optimality principle, which is solved by constructing the Ansatz. The success of this approach is due to the fact that its construction is not restricted to the class of Markovian processes, and thus is applicable to non-Markovian processes. Therefore the difficulties of the non-Markovian characteristic in the Volterra Heston model, which prevents the direct use of the Hamilton-Jacobi-Bellman (HJB) framework, is overcome by this approach.

In addition to this method presented by Han and Wong, we also present an alternative method to solving the Merton optimisation problem, proposed by B\"auerle and Desmettre (see \cite{bauerle_desmettre_2019_portfolio_optimization_des_baulere}). The contrasting difference between these two methods is that B\"auerle and Desmettre motivated by the works in \cite{carmon_et_al_2000} and \cite{harms_stefanovits_2019}, overcome the difficulties of the non-Markovian characteristic of the rough Heston model by means of suitable representation of the fractional part. This is then followed by a suitable finite dimensional approximation, thus allowing the optimisation problem to be cast into the classical framework (see \cite{zariphopoulou_2001_supplement_martingale_distortion_markovian}). Solutions of the Merton optimisation problem are then obtained as the limit of the approximated problem. This approach gives rise to a numerical solution for these types of optimisation problems under RFSV models. In addition to the development of this approach B\"auerle and Desmettre present the approach under a new RFSV model based on the Marchaud fractional derivative, which remedies some of the shortcomings of previous rough Heston models which were derived via means of the fractional derivative.

\section{Organisation of the thesis}

The aim of this thesis is to investigate these two contrasting approaches and compare their corresponding methods and results to the Merton optimisation problem. The outline of the thesis is organised as follows. In chapter 2 we begin by formulating the financial market model and the optimisation problem. This framework casts both of these approaches into a comparable environment. In chapter 3 we introduce the Volterra Heston model and the martingale distortion transformation. This in turn leads to vital pieces of work which motivate the construction of the Ansatz for the martingale optimality principle in these types of problems. The main results are then summarised. In chapter 4 we introduce the approach presented by B\"auerle and Desmettre, and begin by investigating the construction of a rough Heston model via means of the Marchaud fractional derivative and then report the finite dimensional approach. Next a brief outline of the derivation of the HJB equations is given and the key results summarised. Finally we conclude by comparing these two methods and highlight the key assumptions used and future research opportunities.



\chapter{Model definition and preliminaries}
The main goal of this chapter is to introduce the financial market model and the optimisation problem, which will then be solved in the following chapters by considering two different approaches. To define the financial market model we must first begin by introducing the concept of strong and weak solutions of a stochastic differential equation with respect to a Brownian motion. The concept of strong and weak solutions are crucial to the construction of uniqueness and existence arguments for the solution of a stochastic differential equation. Specifically, in chapter 3 we see the case where uniqueness for the stochastic Volterra equations has only been shown in the weak sense, with strong uniqueness still an open problem. Therefore, understanding this difference is of fundamental importance as all processes in the financial market model are defined to be progressively measurable. All notation used in this chapter has been kept consistent with \cite{textbook_karatzas_shreve_stochastic_calcuus}. 

\section{Solutions and uniqueness of the stochastic differential equation}

Let us begin by introducing the definition of a solution of a stochastic differential equation (SDE, for short). For simplicity we develop the following concepts with respect to a diffusion process.

\begin{definition}[A solution of the stochastic differential equation]{\cite[Chapter 5, \S 2]{textbook_karatzas_shreve_stochastic_calcuus}}\label{d-sol-sde}
	Let $\mu : [0,\infty) \times \R^k \to \R^k$ and $\sigma : [0,\infty) \times \R^k \to \R^{k \times d}$ be Borel-measurable functions. Then the stochastic differential equation given by
	\begin{equation}\label{e-sde}
	\begin{split}
	dX_t &= \mu(t,X_t)dt + \sigma(t,X_t)dW_t,\\
	X_0 &= \xi,
	\end{split}	
	\end{equation}
	where $W=\lbrace W_t \sep 0\le t < \infty \rbrace$ is a d-dimensional Brownian motion and $X=\lbrace X_t \sep 0\le t < \infty\rbrace$ with initial condition $X_0 = \xi$, is defined to be the solution of the equation (\ref{e-sde}) and is a suitable stochastic process with continuous sample paths and values in $\R^k$. 
\end{definition}

\noindent This definition of a solution to the stochastic differential equation (\ref{e-sde}) can be given further meaning, in the sense of weak and strong solutions. To understand the crucial differences between these meanings we must first introduce augmentation of $\sigma$-fields. 

\subsection{Augmentation of $\sigma$-fields}
Consider the process $X=\lbrace X_t, \F_t^X \sep 0\le t < \infty\rbrace$, where $\F_t^X =\sigma( X_u \sep 0\le u \le t )$ and $\lbrace \F_t^X \rbrace_{t \ge 0}$ is defined as the \textit{natural filtration} of $X$, and let $\mu$ be the initial distribution of $X_0$ on the space $(\Omega, \F_\infty^X, \P^\mu)$, where $\P^\mu (X_0 \in x) = \mu(x)$ and $x \in \B(\R^d)$. Then define 
\begin{equation}
	\N^\mu \de \left\lbrace N \subseteq \Omega \sep (\exists \, G \in \F_\infty^X) [N \subseteq G \wedge \P^\mu(G)=0]\right\rbrace,
\end{equation}
as \textit{the collection of $\P^\mu$-null sets}. Let us now define the augmentation of $\F_t^X$. 
\begin{definition}[Augmentation of the natural filtration]\cite[Chapter 2, \S 7.2]{textbook_karatzas_shreve_stochastic_calcuus}\label{d-aug}
	\begin{equation}
		\F_t^\mu \de \sigma(\F_t^X \cup \N^\mu), \quad 0\le t <\infty.
	\end{equation}
\end{definition}
\noindent By definition $\lbrace \F_t^\mu \rbrace_{t\ge 0}$ is $\P^\mu$-\textit{complete}, i.e., $\N^\mu \in \F_t^\mu, \fa  0 \le t <\infty$, and we have the following useful results.

\begin{proposition}\label{p-fil-leftc}
	If the process $X$ has left-continuous paths, then the filtration $\left\lbrace \F_t^\mu \right\rbrace_{t\ge 0}$ is left-continuous.
\end{proposition} 

\begin{proof}
	Let $\F_{t-}^\mu=\sigma(\cup_{u<t}\F_u^\mu)$, and $\F_{t-}^X = \sigma(X_u \sep 0 \le u < t)=\F_t^X$, as $\lim_{u\uparrow t}X_u=X_t, \fa 0\le t < \infty$. Thus $\F_{t-}^\mu=\sigma(\cup_{u<t}\F_u^\mu)=\F_t^\mu, \fa 0\le t < \infty$ as required.
\end{proof}

\begin{proposition}\cite[Chapter 2, \S 7.7]{textbook_karatzas_shreve_stochastic_calcuus}\label{p-fil-rightc}
	For a $k$-dimensional strong Markov process $X=\lbrace X_t, \F_t^X \sep 0\le t < \infty \rbrace$ with initial distribution $\mu$, the augmented filtration $\lbrace \F_t^\mu \rbrace_{t\ge 0}$ is right-continuous, meaning 
	$$\F_t^\mu = \cap_{\epsilon>0}\F_{t+\epsilon}.$$
\end{proposition}
\begin{proof}
	See proposition 7.7 in \cite{textbook_karatzas_shreve_stochastic_calcuus}.
\end{proof}

\begin{remark}
	A filtration $\lbrace \F_t \rbrace_{t\ge 0}$ is said to satisfy the \textit{usual conditions} if it's both $\P$-complete and right-continuous.
\end{remark}

\noindent Using the above definitions, let's now define the augmented filtration generated by a Brownian motion.

\begin{definition}[Augmented filtration generated by a Brownian motion]{\cite[Chapter 5, \S 2]{textbook_karatzas_shreve_stochastic_calcuus}}
	Let $W$ be a $d$-dimension Brownian motion $W=\left\lbrace W_t, \F_t^W \sep 0 \le t < \infty \right\rbrace$ defined on the probability space $(\Omega, \F, \P)$, where $\lbrace \F_t^W \rbrace_{t\ge 0}$ is called the natural filtration, as $ \F_t^W \de \lbrace \sigma ( W_u \sep 0\le u \le t ) \rbrace$, and assume that this space also accommodates a random vector $\xi$ taking values in $\R^k$, independent of $\F_\infty^W$. We then consider the filtration
	\begin{equation}
	\G_t \de \sigma(\xi) \vee \F_t^W = \sigma(\xi, W_u \sep 0\le u \le t), \quad 0 \le t < \infty,
	\end{equation}
	which is left-continuous by proposition \ref{p-fil-leftc} and the collection of null sets
	\begin{equation}
	\N \de \left\lbrace N \subseteq \Omega \sep (\exists \, G \in \G_\infty) [N \subseteq G \wedge \P(G)=0]\right\rbrace.
	\end{equation}
	Then the filtration given by
	\begin{equation}\label{e-fil-aug-bm}
	\F_t \de \sigma(\G_t \cup \N), \quad 0\le t < \infty,
	\end{equation}
	is defined as the augmented filtration of the natural filtration $\lbrace \F_t^W \rbrace_{t\ge 0}$.
\end{definition}
\noindent Obviously, $(W_t, \G_t \sep 0\le t < \infty)$ is a $d$-dimension Brownian motion and since the augmentation does not affect the definition of the Brownian motion, then so is $(W_t, \F_t \sep 0\le t < \infty)$. Then, by proposition \ref{p-fil-rightc} the filtration $\lbrace \F_t \rbrace_{t\ge 0}$ satisfies the usual conditions.

Let us now consider the concept of strong solutions for the stochastic differential equation given in (\ref{e-sde}).
\subsection{Strong solutions and uniqueness}

\begin{definition}[A strong solution]{\cite[Chapter 5, \S 2.1]{textbook_karatzas_shreve_stochastic_calcuus}}\label{d-strong-sol}
	A strong solution of the stochastic differential equation (\ref{e-sde}) with initial condition $\xi$, is a process $X=\lbrace X_t \sep 0 \le t < \infty \rbrace$ with continuous sample paths and the following properties:
	\begin{enumerate}[label=(\roman*)]
		\item Let $\lbrace \F_t \rbrace_{t\ge 0}$ be given by (\ref{e-fil-aug-bm}), 
		\item $X=\lbrace X_t, \F_t \sep 0 \le t < \infty \rbrace$ is a continuous, adapted $\R^k$-valued process,\\
		$W=\lbrace W_t, \F_t \sep 0 \le t < \infty \rbrace$ is a $d$-dimensional Brownian motion, and
		\item $\P(\int_{0}^{t}[|\mu_i(u,X_u)| + \sigma_{ij}^2(u,X_u)]du < \infty)=1,\fa 1\le i \le k,\fa 1 \le j \le d,$ and $0\le t < \infty$, $\P$-a.s.. 
	\end{enumerate}
Then, the It\^o process
\begin{equation}
X_t = X_0 + \int_{0}^{t}\mu(u,X_u)du + \int_{0}^{t}\sigma(u,X_u)dW_u, \quad 0\le t < \infty, \, \P\text{-a.s.,}
\end{equation}
is well defined as a strong solution to (\ref{e-sde}).		
\end{definition}

\begin{remark} In some instances as we will see in chapter 3 when we define the stochastic Volterra equations, the stochastic differential equation (\ref{e-sde}) is no longer Markovian. Therefore, the solution is established with the integral form of the SDE pre-constructed. By verifying that the integral form is well defined, meaning it satisfies conditions (i)-(iv), we say it's a strong solution of the stochastic differential equation if it's adapted to the filtration given by (\ref{e-fil-aug-bm}), and a weak solution otherwise (see \cite{jaber_et_al_2017_affine_volterra_processes}).
\end{remark}
\noindent This definition leads naturally to the following definition for strong uniqueness.

\begin{definition}[Strong uniqueness]{\cite[Chapter 5, \S 2.3]{textbook_karatzas_shreve_stochastic_calcuus}}\label{d-strong-uniq} Let $(\Omega, \F, \lbrace\F_t \rbrace_{t\ge 0}, \P)$ be a filtered probability space with $\lbrace \F_t \rbrace_{t\ge 0}$ given by \ref{e-fil-aug-bm}, let $X$ and $\widetilde{X}$ be two strong solutions of \ref{e-sde} relative to any $d$-dimensional Brownian motion $W$ with initial condition $\xi$. Then, we define $X$ as a unique strong solution if
	$$\P(X_t =  \widetilde{X}_t \sep 0 \le t < \infty)=1.$$	
\end{definition}
\noindent From this definition it is clear that the information generated by the initial condition $\xi$ and $\lbrace W_u | 0\le u \le t \rbrace$ determines $X_t$ in an unambiguous way almost everywhere (a.e.), thus satisfying the \textit{principle of causality}. While the conditions of the strong solution provide a nice model framework, it is not always possible to easily prove existence and uniqueness for the strong solution. This naturally leads to the concept of weak solutions, which although weaker than strong solutions are still extremely useful in both theory and applications. 

\subsection{Weak solutions and uniqueness}

\begin{definition}[A weak solution]{\cite[Chapter, \S 3.1]{textbook_karatzas_shreve_stochastic_calcuus}}\label{d-weak-sol} A weak solution of the stochastic differential equation (\ref{e-sde}) is a quadruple $(X, W), (\Omega, \F, \P), \lbrace \F_t \rbrace_{t\ge 0}, (\mu)$, where
	\begin{enumerate}[label=(\roman*)]
		\item $(\Omega, \F, \P)$ is a probability space and $\lbrace \F_t \rbrace_{t\ge 0}$ is a filtration of sub-$\sigma$-fields of $\F$ satisfying the usual conditions, 
		\item $X=\lbrace X_t, \F_t \sep 0 \le t < \infty \rbrace$ is a continuous, adapted, $\R^k$-valued process,\\
		$W=\lbrace W_t, \F_t \sep 0 \le t < \infty \rbrace$ is a $d$-dimensional Brownian motion, and
		(iii), (iv) of definition \ref{d-strong-sol} are satisfied, and
		\item The probability measure $\mu(\Gamma)=\P[X_0 \in \Gamma], \, \Gamma \in \B(\R^k)$ is the initial distribution of the solution.
	\end{enumerate}	
\end{definition}

\begin{remark}
	Clearly the key point of difference between weak and strong solutions is the conditions imposed on filtration $\lbrace \F_t \rbrace_{t\ge 0}$. Whilst we've imposed that the filtration in the strong solution is given by (\ref{e-fil-aug-bm}), meaning the output $X_t$ depends only on the initial condition $\xi$ and the information contained in the filtration $\lbrace \F_t \rbrace_{t\ge 0}$, which in this instance is the set of values $\lbrace W_u | 0\le u \le t \rbrace$; the filtration $\lbrace \F_t \rbrace_{t\ge 0}$ in definition \ref{d-weak-sol} is not restricted to this condition. Thus, in the context of the principle of causality the value of the solution $X_t(\omega)$ is not necessarily given in an unambiguous way by a measurable function of the initial condition $X_0(\omega)$ and the Brownian path $\lbrace W_u(\omega) \sep 0\le u \le t \rbrace$. However, due to constraint imposed in \textit{(ii)}, $W$ is a Brownian motion relative to $\F_t$, which means that the information contained in $X_t(\omega)$ cannot anticipate the future of the Brownian motion. 	
\end{remark}
\noindent The relaxation of this constraint leads to an additional number of approaches in determining weak uniqueness. The following definition of weak uniqueness provides a good example of the added flexibility and is well suited to the concept of weak solutions. 

\begin{definition}[Weak uniqueness]{\cite[Chapter 5, \S 3.4]{textbook_karatzas_shreve_stochastic_calcuus}}\label{d-weak-unique} We say that \textit{uniqueness in the sense of probability law} holds for (\ref{e-sde}) if, for any two weak solutions $(X, W), (\Omega, \F, \P), \lbrace \F_t \rbrace_{t\ge 0}, (\mu)$, and $(\widetilde{X}, \widetilde{W}), (\widetilde{\Omega}, \widetilde{\F}, \widetilde{\P}), \lbrace \widetilde{\F_t} \rbrace_{t\ge 0}, (\widetilde{\mu})$ defined in Definition \ref{d-weak-sol} with the same initial distribution, i.e., 
	$$\mu(\Gamma) = \widetilde{\mu}(\Gamma), \quad \fa \Gamma \in \B(\R^k),$$
have the same law, i.e.,
	$$\P(X\in \Gamma) = \widetilde{\P}(\widetilde{X}\in \Gamma),\quad \fa \Gamma \in \B(\R^k) .$$
\end{definition}

\begin{remark}
	From the above sections it is clear that strong solvability implies weak solvability, however the converse is not always true (see for instance example 5.3.5 from \cite{textbook_karatzas_shreve_stochastic_calcuus}).
\end{remark}

\section{The financial market model and the optimisation problem}\label{sfinancialmarketandoptimisationproblem}

\subsection{The financial market model}\label{sfinancialmarketmodel}
Let $(\Omega, \F, \lbrace \F_t \rbrace_{0 \le t \le T}, \P)$ be a filtered probability space where the filtration $\FF = \lbrace \F_t \rbrace_{0 \le t \le T}$,  supports a two-dimensional Brownian motion $W=\lbrace (W_{1,t}, W_{2,t}), \F_t \sep 0 \le t \le T \rbrace$, and the filtration  $\FF$ is not necessarily the augmented filtration of $W$ given by (\ref{e-fil-aug-bm}), unless stated otherwise. The main reason in considering the financial market model in this setting is due to uniqueness of the stochastic Volterra equation (\ref{estockprocess}) and (\ref{evolterraheston}) to the best of our knowledge having only been proven in the weak sense (see Definition \ref{d-weak-unique}), with its strong uniqueness (see Definition \ref{d-strong-uniq}) currently an open question.

For simplicity we consider a financial market which contains one bond with deterministic bounded risk-free rate $r_t >0$, and one stock or index. The bond evolves according to
\begin{equation}
dS_{0,t} = r_t S_{0,t}dt, \quad S_{0,0}=s_{0,0} > 0,
\end{equation}
whilst the stock price process $S_1=\lbrace S_{1,t}, \F_t \sep 0 \le t \le T \rbrace$ evolves according to
\begin{equation}\label{estockprocess}
dS_{1,t} = S_{1,t}(r_t + \theta V_t)dt + S_{1,t} \sqrt{V_t}dW_{1,t}, \quad S_{1,0}= s_{1,0} > 0.
\end{equation}
and is a continuous, adapted, $\R_+$-valued process with $\theta \ne 0$, and $V=\lbrace V_t, \F_t \sep 0 \le t \le T \rbrace$ the variance process, which is also a continuous, adapted, $\R_+$-valued process. The variance process will be given by two different rough volatility models defined in chapters 3 and 4. In addition, we construct the Brownian motion $B=\lbrace B_t, \F_t \sep 0\le t \le T \rbrace$ which will be used to drive the variance process $V$, and is given by 
\begin{equation}\label{eBt}
	dB_t = \rho dW_{1,t} + \sqrt{1-\rho^2}dW_{2,t}, \quad  \rho \in (-1,1).
\end{equation}

\subsection{The optimisation problem}\label{sub-optimisation-problem}

The optimisation problem is to find investment strategies in the financial market that
maximise the expected utility from terminal wealth. The main reason for choosing to investigate the optimisation problem in the context of a utility function is that it creates a flexible and tractable way for defining risk-aversion. Specifically, we choose to investigate the optimisation problem under a constant relative risk aversion (CRRA) utility function given by the power utility function
\begin{equation}\label{eutility}
	U(x) \de \frac{1}{\gamma}x^\gamma, \quad 0 < \gamma < 1.
\end{equation}
The parameter $\gamma$ in (\ref{eutility}) represents the risk aversion of the investor, with smaller (resp. higher) values of $\gamma$ corresponding to higher (resp. lower) risk aversion of the investor. 

Let $\pi_t \in \R$ be the fraction of wealth invested in the stock, and $1-\pi_t$ the fraction of wealth
invested in the bond at time $t$. This allows for the inclusion of shorting a position in the stock ($\pi_t < 0, \, t\in [0,T]$), and shorting the bond position (taking a loan) to acquire a credit position in the stock ($\pi_t >1, \, t\in [0,T]$), thus  $\pi=(\pi_t)_{0\le t \le T}$ is defined as the investment strategy. Moreover, an admissible investment
strategy must be an $\FF$-adapted process such that all integrals exist. The conditions of admissibility are made clear in Definition \ref{dadmissible}. 

Let us denote the self-financing wealth process as $\Pi^\pi=\lbrace \Pi^\pi_t, \F_t |t\in [0,T]\rbrace$, under an admissible investment strategy $\pi$, which evolves according to the stochastic differential equation

\begin{equation}\label{ewealthprocess}
	d\Pi^\pi_t = \Pi^\pi_t(r_t + \theta \pi_t V_t) dt + \Pi^\pi_t\pi_t \sqrt{V_t}  dW_{1,t}, \quad \Pi^\pi_0 = w_0 >0.
\end{equation}

\noindent Now, as we wish to optimise our investment strategy $\pi$, we can see that the optimal choice of $\pi$ depends implicitly through the solution $\Pi^\pi$. This leads naturally to the following uniqueness definition of (\ref{ewealthprocess}).

\begin{definition}[Solution to the stochastic differential equation with random coefficients and its uniqueness]{\cite[Chapter 1, \S 6.4]{textbook_Yong_Zhou_1999_stochastic_controls}}\label{d-sde-random} 
	Let $\mu : [0,\infty) \times \R^k \times \Omega \to \R^k$ and $\sigma : [0,\infty) \times \R^k \times \Omega  \to \R^{k \times d}$ be given on a given filtered probability space $(\Omega, \F, \lbrace \F_t \rbrace_{t\ge 0}, \P )$, which accommodates a $d$-dimensional Brownian motion $W=\lbrace W_t, \F_t \sep 0 \le t < \infty \rbrace$ and a random vector $\xi$ taking values in $\R^k$, independent of $\F_\infty^W$ and $\F_0$-measurable. An $\lbrace \F_t \rbrace_{t\ge 0}$-adapted continuous process $(X_t)_{t\ge 0}$, is called a solution of 
	\begin{equation}\label{e-sde-random}
	\begin{split}
	dX_t(\omega) &= \mu(t, X_t, \omega)dt + \sigma(t,X_t,\omega)dW_t(\omega), \\
	X_0(\omega) &= \xi(\omega),
	\end{split}	
	\end{equation}
	if
	\begin{enumerate}[label=(\roman*)]
		\item $X_0 = \xi, \quad \P$-a.s.,
		\item  $\int_{0}^{t}(|\mu(u, X_u(\omega), \omega)| + |\sigma(u, X_u(\omega),\omega)|^2)du < \infty, \quad 0\le t < \infty,$  $\P$-a.s. $\forall \; \omega \in \Omega$, and
		\item $X_t(\omega) = \xi(\omega) + \int_{0}^{t}\mu(u, X_u(\omega),\omega)du + \int_{0}^{t}\sigma(u,X_u(\omega),\omega)dW_u(\omega),\quad 0\le t < \infty,$ $\P$-a.s. $\forall \; \omega \in \Omega$.
	\end{enumerate}
And we say that the solution is unique if
	\begin{equation}
	\P(X_t = Y_t \sep 0 \le t < \infty)=1,
	\end{equation}
	holds for any two solutions $X$ and $Y$.
\end{definition}

\noindent This leads to the natural development of our admissibility conditions of an admissible investment strategy.

\begin{definition}[Admissible investment strategy]\cite[\S 2.3]{han_wong_portfolio_utility_volterra_heston}\label{dadmissible} 
	An investment strategy $\pi$ is said to be admissible if
	\begin{enumerate}[label=(\roman*)]
		\item $\pi$ is $\FF$-adapted, and $\int_{0}^{T}\pi_u^2V_u du < \infty, \, \P$-a.s. 
		\item  the wealth process \ref{ewealthprocess} has a unique solution in the sense of Definition \ref{d-sde-random},
		\item $\Pi^\pi_t \ge 0$, $t\in [0,T]$, $\P$-a.s., and
		\item $\E[(\Pi^\pi_T)^{p\gamma}]<\infty$, for some $p>1$.
	\end{enumerate}	
The set of all admissible investment strategies is denoted as $\A$.
\end{definition}

\noindent Then the optimal investment strategy of the wealth process (\ref{ewealthprocess}), is given by

\begin{equation}\label{eoptimisation}
\J_0^* \de \underset{\pi\in \A}{\sup}\,\E\left[\frac{(\Pi^\pi_T)^\gamma}{\gamma}\right],
\end{equation}

\noindent or equivalently as

\begin{equation}\label{eoptimisationlocal}
\J_t^* \de \underset{\pi\in \A}{\sup}\,\E\left[\frac{(\Pi^\pi_T)^\gamma}{\gamma}\bigg|\F_t\right]. 
\end{equation}

\noindent Now, before commencing the next chapter we briefly give some motivation on how to solve (\ref{eoptimisation}) with respect to the utility function (\ref{eutility}).


\subsection{The distortion transformation approach to solving the optimisation problem}\label{soptimisationsolution}
The investor's goal is to find the optimal admissible investment strategy which we denote as $\pi^*$, so as to maximise their expected utility of terminal wealth (\ref{eoptimisation}). In 2001 Zariphopoulou (see \cite{zariphopoulou_2001_supplement_martingale_distortion_markovian}) applied the dynamic programming approach (see Theorem \ref{tdpp}) to study this optimisation problem with respect to the power utility under a Markovian model. Interestingly, the maximum expected utility (\ref{eoptimisationlocal}) is of the form $(\E\xi^{1/\delta})^\delta$ for a random variable $\xi$ depending on the variance process and the market parameters, with $\delta$ a constant commonly referred to as the distortion power. Since this discovery a number of papers (see \cite{bauerle_desmettre_2019_portfolio_optimization_des_baulere},\cite{foque_hu_2018a_optimal_portfolio_fast_smooth}, \cite{foque_hu_2018b_optimal_portfolio_fast_rough}, \cite{foque_hu_2019_optimal_portfolio_fractional_martingale_distortion}, \cite{han_wong_portfolio_mean-variance_volterra_heston} and \cite{han_wong_portfolio_utility_volterra_heston}) have utilised this structure. Let us now recall the distortion transformation obtained by Zariphopoulou (see \cite{zariphopoulou_2001_supplement_martingale_distortion_markovian}), as this work motivates the two different approaches taken to solve the optimisation problem in chapters 3 and 4.

Let us consider the Markovian setup under the complete filtered probability space $(\Omega, \F, \lbrace\F_t\rbrace_{t\in [0,T]}, \P)$ satisfying the usual conditions. Let's then define the process $\Pi^\pi=\lbrace \Pi^\pi_t, \F_t |t\in [0,T]\rbrace$ to be a self-financing wealth process under an admissible investment strategy $\pi$, which evolves according to the stochastic differential equation	
\begin{equation}\label{esdewealthdef}
d\Pi^\pi_t = \Pi^\pi_t\pi_t\mu(Y_t) dt + \Pi^\pi_t\pi_t\sigma(Y_t) dW_t^\Pi, \quad \Pi^\pi_0 = w_0 >0,
\end{equation}

\noindent where process $Y=\lbrace Y_{t}, \F_t |t\in [0,T]\rbrace$ has dynamics given by
\begin{equation}\label{esdewealthprocess}
dY_t = a(Y_t) dt + b(Y_t) dW_t^Y, \quad Y_0 = y_0 \in \R.
\end{equation}

\noindent The processes $W^\Pi$ and $W^Y$ are Brownian motions on the probability space and correlated with coefficient $\rho \in (-1,1)$. The coefficients $\mu, \sigma, b,a$ are functions of the factor $Y$ and time and are assumed to satisfy all the required regularity assumptions in order to guarantee that a unique solution to (\ref{esdewealthdef})-(\ref{esdewealthprocess}) exists. Therefore the process $X=(\Pi^\pi,Y)$ is a diffusion process, and by extension a Markov process. 

\begin{proposition}[Distortion transformation]\cite[Proposition 2.1, Section 2.1]{zariphopoulou_2001_supplement_martingale_distortion_markovian, foque_hu_2019_optimal_portfolio_fractional_martingale_distortion}\label{pdistortiontransform}The value function $\J_t^*$ is given by
	\begin{equation}\label{edistortion}
	\J_t^* \de \underset{\pi\in \A}{\sup}\,\E \left[\frac{(\Pi^\pi_T)^\gamma}{\gamma} \bigg|\Pi^\pi_t=w, \, Y_t=y\right] =\frac{w^\gamma}{\gamma}\Psi(t,y)^\delta,
	\end{equation}
	
	\noindent with 
	\begin{equation}\label{edistortionpower}
	\delta = \frac{1-\gamma}{1-\gamma + \gamma \rho^2},
	\end{equation}
	
	\noindent which results in cancelling $(\Psi_y)^2$ terms in the HJB equation (see \cite{foque_hu_2019_optimal_portfolio_fractional_martingale_distortion}). Therefore, $\Psi$ solves the linear PDE
	\begin{equation}\label{edistortionpdesol}
	\Psi_t + \left(\frac{b^2(y)}{2}\partial_{yy} + a(y)\partial_y + \frac{\gamma\lambda(y)\rho b(y)}{1-\gamma}\partial_y\right)\Psi + \frac{\gamma\lambda^2(y)}{2\delta(1-\gamma)}\Psi=0, \quad \Psi(T,y)=1,
	\end{equation}

	\noindent where $\lambda(y)\de \mu(y)/\sigma(y)$ is the Sharpe ratio.
\end{proposition}

\begin{proof}
	See Proposition 2.1 in \cite{zariphopoulou_2001_supplement_martingale_distortion_markovian} and Section 2.1 in \cite{foque_hu_2019_optimal_portfolio_fractional_martingale_distortion}.
\end{proof}

 \noindent Recognising that the solution to (\ref{edistortionpdesol}) is given by the Feynman-Kac representation (see \cite{textbook_pham}), we have
\begin{equation}\label{edistortionsolution}
\Psi(t,y)= \widetilde{\E}\left[\exp \left\lbrace \frac{\gamma}{2\delta(1-\gamma)}\int_{t}^{T}\lambda^2(Y_u)du\right\rbrace \bigg | Y_t=y\right],
\end{equation}

\noindent where we define the probability measure $\widetilde{\P}$, such that 

\begin{equation}\label{edistortionbrownian}
\widetilde{W}_t^Y=W_t^Y-\frac{\rho \gamma}{1-\gamma}\int_{0}^{t}\lambda(Y_u)du
\end{equation}
\noindent is a Brownian motion under $\widetilde{\P}$.

From these results the next chapters look at two different approaches on how to solve this optimisation problem. Specifically, the focus of chapter 4 is on the use of the classical method which solves the optimisation problem through the use of PDE arguments, as shown in the above definition under a rough Heston model. This is achieved for the non-Markovian model through the use of a finite dimensional approximation method proposed in \cite{bauerle_desmettre_2019_portfolio_optimization_des_baulere}, which allows the problem to be cast into the classical Markovian optimisation framework. In the next chapter a different approach is taken, namely, the martingale distortion transformation which utilises the martingale optimality principle (see \cite{textbook_pham}) under a Volterra Heston model. As we will soon show, both of these methods have their associated advantages and disadvantages, with the goal of this thesis being to highlight this and provide future research questions.


\chapter{A martingale distortion approach to solving the portfolio optimisation problem under the Volterra Heston model}\label{nlmca}

In this chapter we begin with a brief introduction to the concept of stochastic calculus of convolutions and resolvents, which leads naturally to the introduction of Affine Volterra processes developed in \cite{textbook_gustaf_et_al_1990_volterra} and \cite{jaber_et_al_2017_affine_volterra_processes}. A special case of the Affine Volterra process is then considered, focusing on the Volterra Heston model used in \cite{han_wong_portfolio_utility_volterra_heston}. We then consider the derivation of the optimal trading strategy under the Volterra Heston model in \cite{han_wong_portfolio_utility_volterra_heston}. The chapter concludes with a brief literature review on the martingale distortion transformation, which leads to the motivation behind the construction of the Ansatz in \cite{han_wong_portfolio_utility_volterra_heston} for the martingale optimality principle (see Definition \ref{dmartingaleoptimalityp}).  All notation is kept consistent with chapter 2 and \cite{textbook_gustaf_et_al_1990_volterra}. We also adopt the convention used in \cite{jaber_et_al_2017_affine_volterra_processes} to differentiate between a row and column vector through the use of an asterisk (i.e. $(\C^k)^*$ (resp. $\C^k$) is a $k$-dimensional row (resp. column) vector of complex numbers).


\section{Stochastic calculus of convolutions and resolvents}

\subsection{Convolutions and some associated properties}
\begin{definition}[Convolution of two functions]\cite[Chapter 2, \S 2.1]{textbook_gustaf_et_al_1990_volterra}\label{d-conv-ff} The convolution $K*F$ of two functions $K$ and $F$ defined on $\R_+$ is the function
	\begin{equation}\label{e-convff}
	(K * F)(t) = \int_{0}^{t}K(t-u)F(u)du, 
	\end{equation}
	which is defined $\fa t \in R_+$ for which the integral exists.
\end{definition}

\noindent Moreover, (\ref{e-convff}) has a number of useful algebraic properties, which can be obtained from a change of variables and Fubini's Theorem. For example, (\ref{e-convff}) is commutative and associative (see Chapter 2 in \cite{textbook_gustaf_et_al_1990_volterra}, for further properties and details).

\noindent The following definition has stricter meaning, (i.e. order does not change the definition of the convolution).

\begin{definition}[Convolution of a measurable function and measure]\cite[Chapter 3, \S 2.1]{textbook_gustaf_et_al_1990_volterra}\label{d-conv-fmeasure} Let $K$ be a measurable function on $\R_+$, and $L$ be a measure on $\R_+$ of locally bounded variation. Then, the convolutions $K*L$ and $L*K$ are defined by
	\begin{equation}\label{e-conv-fmart}
	(K * L)(t) = \int_{0}^{t}K(t-u)L(du), \quad (L*K)(t)=\int_{0}^{t}L(du)K(t-u)
	\end{equation}
	for $t>0$ under proper conditions, and extended to $t=0$ by right-continuity when possible.
\end{definition}

\noindent Therefore (\ref{e-conv-fmart}) is by definition associative. The last convolution that we introduce is the convolution of a measurable function and local martingale.

\begin{definition}[Convolution of a measurable function and local martingale]\cite[\S 2.3]{jaber_et_al_2017_affine_volterra_processes}\label{d-conv-fmartingale} Let $M$ be a $d$-dimensional continuous local martingale defined on the probability space $(\Omega, \F, \P)$ and $K$ a function $K: \R_+ \to \R^{k \times d}$. Then the convolution of $K$ and $M$ is given by 
	\begin{equation}\label{e-conv-fmartingale}
	(K * dM)_t = \int_{0}^{t}K(t-u)dM_u, 
	\end{equation}
	and is well-defined as an It\^o integral if 
	$$\int_{0}^{t}|K(t-u)|^2d\text{tr}\langle M \rangle_u < \infty, \quad 0\le t < \infty, \text{ $\P$-a.s..}$$
\end{definition}

\begin{proposition}\cite[\S 2]{jaber_et_al_2017_affine_volterra_processes}
	If $K\in L^2_\text{loc}(\R_+; \R^{k \times d})$ and $\langle M \rangle_u = \int_{0}^{u}a_sds$ for some locally bounded process $a$. Then, (\ref{e-conv-fmartingale}) is well defined $\fa t \ge 0$.
\end{proposition}
\begin{proof}
	\begin{align*}
		\int_{0}^{t}|K(t-u)|^2d\text{tr}\langle M \rangle_u \le \underset{0\le s \le t}{\max}|tr(a_s)|\|K\|_{L^2([0,t])}^2<\infty,
	\end{align*}
	by definition.
\end{proof}

\noindent The next result provides a useful relationship between (\ref{e-conv-fmartingale}) and (\ref{e-conv-fmart}). Let's firstly introduce an important theorem which is relevant to a number of proofs in this chapter. 

\begin{theorem}[Stochastic Fubini Theorem]\cite[Theorem 2.2]{veraar_2012_stochastic_fubini}\label{theoremstochasticfubini}
	Let $(\Omega, \F, \lbrace \F_t\rbrace_{t\ge0}, \P)$ be a complete filtered probability space satisfying the usual conditions, and $M$ a continuous local martingale. Let $X=\lbrace X_t,\F_t \sep t\ge0\rbrace$ be a collection of $(\R, \B(\R),\mu)$-valued random variables. Let $\psi:[0,T]\times\R \times \Omega \to \R$ be progressively measurable and such that for almost all $\omega \in \Omega$,
	\begin{equation}
	\int_\R\left(\int_{0}^{T}|\psi(t,x,\omega)|^2d\langle M\rangle(t,\omega) \right)^{1/2}d\mu(x) < \infty,
	\end{equation}
	then
	\begin{equation}
	\int_\R\left(\int_{0}^{T}\psi(t,x,\omega)d M(t,\omega) \right)d\mu(x) =\int_{0}^{T}\left(\int_\R\psi(t,x,\omega)d\mu(x) \right)d M(t,\omega).
	\end{equation}
\end{theorem}

\begin{remark}
	The stochastic Fubini Theorem in \cite{veraar_2012_stochastic_fubini} can be extended to semimartingales under additional constraints. However, for the purpose of this thesis it is only necessary to consider the stochastic Fubini Theorem with respect to a martingale.
\end{remark}

\noindent Using Theorem \ref{theoremstochasticfubini} the associative property of (\ref{e-conv-fmartingale}) is obtained (see \cite{jaber_et_al_2017_affine_volterra_processes}).

\begin{proposition}\cite[Lemma 2.1]{jaber_et_al_2017_affine_volterra_processes}
	Let $K\in L^2_\text{loc}(\R_+; \R^{k \times d})$, $L$ be a $\R^{n\times k}$-valued measure on $\R_+$ of locally bounded variation and $M$ be a $d$-dimensional continuous local martingale, with $\langle M \rangle_u = \int_{0}^{u}a_sds, \, t\ge 0$ for some locally bounded process $a$. Then
	\begin{equation}
		(L*(K*dM))_t = ((L*K)*dM)_t,
	\end{equation}
	$\fa t\ge 0$. If $F\in L_\text{loc}^1(\R_+;\R^{n\times m})$, then setting $L(dt)=Fdt$ we obtain,
	$$(F*(K*dM))_t = ((F*K)*dM)_t.$$
\end{proposition}

\noindent This result has been proven in \cite{jaber_et_al_2017_affine_volterra_processes}, however we reformulate to give the reader a taste in the usefulness of the stochastic Fubini Theorem.

\begin{proof} To ease notation let $L=L_{ij}, K=K_{jl}$, and $M=M_l$, then using (\ref{e-conv-fmart}) and (\ref{e-conv-fmartingale}) we have,
	\begin{align*}
	(L*(K*dM))_t &= \int_{0}^{t}(K*dM)(t-u)L(du),\\
	&=  \int_{0}^{t}\left(\int_{0}^{t-u}K(t-u-s)dM_s\right)L(du),\\
	&=  \int_{0}^{t}\left(\int_{0}^{t}\textbf{1}_{\lbrace u<t-s \rbrace}K(t-u-s)dM_s\right)L(du).
	\end{align*} 
	As 
	$$\int_{0}^{t}\left(\int_{0}^{t}\textbf{1}_{\lbrace u<t-s \rbrace}|K(t-u-s)|^2d\langle M \rangle _s\right)^{1/2}L(du)\le \underset{0\le u \le t}{\max}|a_s|^{1/2}\|K\|_{L^2([0,t])}L([0,t]),$$
	is finite $\P$-a.s., then Theorem \ref{theoremstochasticfubini} yields,
	$$(L*(K*dM))_t\int_{0}^{t}\left(\int_{0}^{t}\textbf{1}_{\lbrace u<t-s \rbrace}K(t-u-s)L(du)\right)dM_s=((L*K)*dM)_t.$$
	As this holds $\fa 1\le i\le n, 1\le j \le k, 1\le l \le d$, the desired result is obtained.	
\end{proof}

\noindent We now give a following important assumption on the kernel $K$ which helps derive a number of important results.

\begin{assumption}\cite[Assumption 2.5]{jaber_et_al_2017_affine_volterra_processes}\label{assump1}
	Assume, the kernel $K\in L^2_\text{loc}(\R_+;\R)$ and there exists $\gamma \in (0,2]$ such that $\int_{0}^{h}K(t)^2dt = O(h^\gamma)$ and $\int_{0}^{T}(K(t+h)-K(t))^2dt = O(h^\gamma), \fa 0\le T<\infty$.
\end{assumption}

\noindent Example 2.3 in \cite{jaber_et_al_2017_affine_volterra_processes}, lists a number of kernels which satisfy Assumption \ref{assump1}. Under Assumption \ref{assump1} we have the following useful result.

\begin{proposition}\cite[Lemma 2.4]{jaber_et_al_2017_affine_volterra_processes}\label{plemma24}
	Assume $K$ satisfies Assumption \ref{assump1} and consider a process $X=\lbrace X_t, \F_t \sep 0\le t \le T \rbrace$ defined on the filtered probability space $(\Omega, \F, \lbrace \F_t \rbrace_{0\le t \le T}, \P)$, with $X_t=K*(b dt + dM)_t$, where $\mu$ is an $\lbrace \F_t \rbrace_{0\le t \le T}$-adapted process and $M$ is a continuous local martingale, with $\langle M \rangle_u = \int_{0}^{u}a_sds, \, t\ge 0$ for some $\lbrace \F_t \rbrace_{0\le t \le T}$-adapted process $a$. Let $0\le T<\infty$, and $p>2/\gamma$, be such that $\underset{0 \le t\le T}{\sup}(\E[|a_t|^{p/2}+|b_t|^p])<\infty$. Then, $X$ admits a version which is H\"{o}lder continuous
	\begin{equation}
		|X_t - X_s|\le C | t - s |^\alpha, \quad  \alpha < \gamma/2 - 1/p, \fa t,s\in [0,T]	,
	\end{equation}
	and
	\begin{equation}
	\E \left[\left(|X|_{C^{0,\alpha}([0,T])}\right)^p\right] \de \E \left[\left(\underset{0 \le s<t\le T}{\sup}\frac{|X_t-X_s|}{|t-s|^\alpha}\right)^p\right]\le c \underset{0 \le t\le T}{\sup}(\E[|a_t|^{p/2}+|b_t|^p]),
	\end{equation}
	$\fa \alpha \in [0, \gamma/2 - 1/p)$, where $c$ is a constant that only depends on $p, K,$ and $T$. Moreover, if $a$ and $b$ are locally bounded, then $X$ admits a version which is H\"{o}lder continuous of any order $\alpha<\gamma/2$.
\end{proposition}
\begin{proof} See Lemma 2.4 in \cite{jaber_et_al_2017_affine_volterra_processes}.
\end{proof}
 
\subsection{Resolvents and some examples}
Let us now provide a brief introduction to resolvents, as their definitions are used throughout.
\begin{definition}[Resolvent of the first kind]\cite[Chapter 5, \S 5.1]{textbook_gustaf_et_al_1990_volterra, jaber_et_al_2017_affine_volterra_processes}\label{dres1} Let $K \in L_\text{loc}^1(\R^+; \R^{d\times d})$ and the locally bounded measure $L\in M_\text{loc}^1(\R^+; \R^{d\times d})$ satisfy
	\begin{equation}
		K*L = L*K \equiv id,
	\end{equation}
	where $id$ is the $d$-dimension identity matrix. Then $L$ is said to be resolvent of the first kind of $K$.
\end{definition}

\begin{definition}[Resolvent of the second kind]\cite[\S 2.11]{jaber_et_al_2017_affine_volterra_processes}\label{dres2}	Let $K \in L_\text{loc}^1(\R^+; \R^{d\times d})$ and $R \in L_\text{loc}^1(\R^+; \R^{d\times d})$ satisfy
	\begin{equation}
	K*R = R*K = K-R,
	\end{equation}
	then $R$ is said to be resolvent, or resolvent of the second kind of $K$.
\end{definition}

\noindent Further properties of these definitions can be found in \cite{textbook_gustaf_et_al_1990_volterra} and
\cite{jaber_et_al_2017_affine_volterra_processes}. Commonly used kernels are
summarised in the table below, which is also available in \cite{jaber_et_al_2017_affine_volterra_processes}.

\begin{table}[h]\label{table_kernels}
	\begin{tabular}{c c c c} 
		\hline
		 & $K(t)$  & $L(dt)$ & $R(t)$ \\ \hline 
		
		\rule{0pt}{1.2\normalbaselineskip} Constant & c  & $c^{-1}\delta_0(dt)$ & $ce^{-ct}$\\ 
		
		\rule{0pt}{1.2\normalbaselineskip} Fractional & $c\frac{t^{\alpha-1}}{\Gamma(\alpha)}$  & $c^{-1}\frac{t^{-\alpha}}{\Gamma(1-\alpha)}dt$ & $ct^{\alpha-1}E_{\alpha,\alpha}(-ct^\alpha)$\\
		
		\rule{0pt}{1.2\normalbaselineskip} Exponential & $ce^{-\lambda t}$  & $c^{-1}(\delta_0(dt)+\lambda dt)$ & $ce^{-\lambda t}e^{-ct}$\\	
		
		\rule{0pt}{1.2\normalbaselineskip} Gamma & $ce^{-\lambda t}\frac{t^{-\alpha-1}}{\Gamma(\alpha)}$  & $c^{-1}\frac{1}{\Gamma(1-\alpha)}e^{-\lambda t}\frac{d}{dt}(t^{-\alpha}*e^{\lambda t})(t)dt$ & $ce^{-\lambda t}t^{-\alpha-1}E_{\alpha,\alpha}(-ct^\alpha)$\\	[1ex] \hline 
		\rule{0pt}{.5\normalbaselineskip}	
	\end{tabular}
\caption{Examples of kernels $K$ and their corresponding first and second kind resolvents $L$ and $R$ respectively,  where $E_{\alpha, \beta}(u)=\sum_{n=0}^{\infty}\frac{u^n}{\Gamma(\alpha n + \beta)}$ is the Mittag-Leffler function, some of its associated properties can be found in \cite{euch_rosenbaum_2019_characteric_function_rough_heston_models}.}
\end{table} 
\vskip -0.5cm

\noindent These definitions outlined provide the necessary building blocks to give meaningful definition to the Volterra Heston model. In the next section affine Volterra processes are introduced, and as we will see, these processes provide a useful framework to establish a number of models, with the Volterra Heston model being a specific case.


\section{Affine Volterra processes}

Fix $k \in \NN$ and $K\in L_\text{loc}^2(\R_+;\R^{k\times k})$, and let $a:\R^k \to \R^{k\times k}$ and $b:\R^k \to \R^k$ be affine maps given by
\begin{equation}\label{e-affine-maps}
\begin{split}
a(x_1, x_2, \dots, x_k) &= A^0 + A^1x_1 + \dots + A^kx_k,\\
b(x_1, x_2, \dotsc, x_k) &= b^0 + x_1b^1 + \dots + x_kb^k,
\end{split}
\end{equation}
where $A^i \in \R^{k\times k}$ are $k$-dimension symmetric matrices and $b\in \R^k, \, i=0,1,\dots,k$. To ease notation define the $k \times k$ matrix 
\begin{equation}
	B= \begin{bmatrix} b^1 & b^2 & \dots & b^k\end{bmatrix},
\end{equation} 
and the row vector
\begin{equation}
A(\nu) = \begin{bmatrix} \nu A^1\nu^T & \nu A^2\nu^T & \dots & \nu A^k\nu^T\end{bmatrix},
\end{equation}
where $\nu\in (\C^k)^*$ is a row vector. Let us denote $S$ as the state space defined on $\R^k$ and assume that $a(x)$ is positive semidefinite $\fa x \in S$. Then, $a(x)$ admits the following decomposition: let $\sigma : \R^k \to \R^{k \times d}$ be continuous and satisfy 
\begin{equation}\label{easigdecomposition}
\sigma(x) \sigma(x)^T=a(x) \fa x\in S.
\end{equation}

\noindent This construction naturally leads to the following definition of an affine Volterra process (see \cite{jaber_et_al_2017_affine_volterra_processes}).

\begin{definition}[Affine Volterra process]\cite[Definition 4.1]{jaber_et_al_2017_affine_volterra_processes}\label{d-affine-vol}
 An affine Volterra process defined on the probability space $(\Omega, \F_t, \P)$ is a continuous $S$-valued solution $X=\lbrace X_t,\F_t \sep 0\le t < \infty\rbrace$ of
 \begin{equation}\label{affinevolterraprocess}
 	X_t = X_0 + \int_{0}^{t}K(t-u)b(X_u)du + \int_{0}^{t}K(t-u)\sigma(X_u)dW_u,
 \end{equation}
 with $a=\sigma \sigma^T$ and $b$ given by (\ref{e-affine-maps}), $W=\lbrace W_t, \F_t \sep 0\le t < \infty \rbrace$ is a $d$-dimensional Brownian motion, and $X_0$ is a deterministic value in $S$. 
\end{definition}

\noindent To ease notation, we can write (\ref{affinevolterraprocess}) more compactly as
$$X=X_0 + K*(b(X)dt + \sigma(X)dW).$$
Moment bounds for any solution under definition \ref{d-affine-vol} is given by the following result in \cite{jaber_et_al_2017_affine_volterra_processes}.

\begin{proposition}\cite[Lemma 3.1]{jaber_et_al_2017_affine_volterra_processes}\label{plemma31} Assume $b(\cdot)$ and $\sigma(\cdot)$ are continuous and satisfy the linear growth condition
	\begin{equation}
	|b(x)| + |\sigma(x)|\le c_{LG}(1+|x|), \quad x\in \R^k,
	\end{equation}
	for some constant $c_{LG}$. Let $X$ be a continuous solution of (\ref{affinevolterraprocess}). Then for any $p\ge2$ and $T<\infty$ one has
	$$\underset{t\le T}{\sup}\E[|X_t|^p]\le c,$$
	for some constant $c$ that depends only on $|X_0|,\, K|_{[0,T]},\, c_{LG}, \,p$ and $T$.
\end{proposition}

\begin{proof}
	See lemma 3.1 in \cite{jaber_et_al_2017_affine_volterra_processes}.
\end{proof}

\begin{remark}\cite[Remark 3.2]{jaber_et_al_2017_affine_volterra_processes}\label{r32}
	From the construction of the proof of Proposition \ref{plemma31} in \cite{jaber_et_al_2017_affine_volterra_processes}, we have that Proposition \ref{plemma31} also holds for state and time-dependent predictable coefficients $$|b(t,x,\omega)|+|\sigma(t,x,\omega)|\le c_\text{LG}(1+|x|), \quad x\in \R^k, \, t\in \R_+, \, \omega \in \Omega,$$ 
	for some constant $c_\text{LG}$.
\end{remark}

\noindent We conclude this section with the following powerful theorem. Theorem 4.3 in \cite{jaber_et_al_2017_affine_volterra_processes} is used extensively in the construction of a number of proofs. Moreover, we will see the importance this result has in the construction of the Ansatz in the martingale optimality principle (see Section \ref{sconstructionansatz}).

\begin{theorem}\cite[Theorem 4.3]{jaber_et_al_2017_affine_volterra_processes}\label{t43} Let $X$ be an affine Volterra process, as in definition \ref{d-affine-vol}, and let $T<\infty$, $\nu \in (\C^k)^*$, and $f\in L^1([0,T];(\C^k)^*)$. Assume $\varphi \in L^2([0,T]; (\C^k)^*)$ solves the Riccati-Volterra equation
	\begin{equation}\label{ericcativolterra43}
	\varphi = \nu K + (f + \varphi B + \frac{1}{2}A(\varphi))*K.
	\end{equation}
	Then the process $Y=\lbrace Y_t \sep 0\le t \le T\rbrace$ defined by
	\begin{equation}\label{e431}
		Y_t = Y_0 + \int_{0}^{t}\varphi(T-u)\sigma(X_u)dW_u - \frac{1}{2}\int_{0}^{t}\varphi(T-u)a(X_u)\varphi(T-u)^Tdu,
	\end{equation}
	and	
	\begin{equation}\label{e432}
	Y_0 = \nu X_0 + \int_{0}^{T}\left(f(u)X_0 + \varphi(u)b(X_0)+\frac{1}{2}\varphi(u)a(X_0)\varphi(u)^T\right)du,
	\end{equation}
	satisfies
	\begin{equation}\label{e433}
	Y_t = \E[\nu X_T + (f*X)_T|\F_t] + \frac{1}{2}\int_{t}^{T}\varphi(T-u)a(\E[X_u|\F_t])\varphi(T-u)^Tdu, 
	\end{equation}
	$\fa t \in [0,T]$. Then, the process $\exp(Y)=\lbrace \exp(Y_t) \sep 0 \le t \le T\rbrace$ is a local martingale, and if it is a true martingale, $\exp(Y)$ has the following exponential-affine transform representation
	\begin{equation}\label{e434}
	\E[\exp(\nu X_T + (f*X)_T)|\F_t]=\exp(Y_t), \quad 0\le t \le T.
	\end{equation} 
\end{theorem}

\noindent Using these results we can now introduce and give a meaningful definition to the \textit{Volterra Heston model} under the financial market model defined in Section \ref{sfinancialmarketandoptimisationproblem}. 


\section{The Volterra Heston model}

Under the financial market model defined in Section \ref{sfinancialmarketandoptimisationproblem}, let the variance process  $V$ be given by
\begin{equation}\label{evolterraheston}
V_t = v_0 + \kappa \int_{0}^{t}K(t-u)(\phi-V_u)du + \int_{0}^{t}K(t-u)\sigma \sqrt{V_u}dB_u,
\end{equation}
where the kernel $K$ satisfies Assumption \ref{assump2}, $B$ is defined by (\ref{eBt}), and parameters $(v_0, \kappa, \phi, \sigma)\in \R_+$. Recall that the stock price process $S_1$ defined by (\ref{estockprocess}), reads
$$dS_{1,t} = S_{1,t}(r_t + \theta V_t)dt + S_{1,t} \sqrt{V_t}dW_{1,t}, \quad S_{1,0} = s_{1,0} > 0,$$
and since the process $A_t=\int_{0}^{t}\left(r_u+\theta V_u-\frac{1}{2}V_u^2\right)du + \int_{0}^{t}\sqrt{V_u}dW_{1,u}$ is a semimartingale, by It\^o's lemma the log-price satisfies
$$\log S_{1,t} = \log S_{1,0} + \int_{0}^{t}\left(r_u+\theta V_u-\frac{1}{2}V_u\right)du + \int_{0}^{t}\sqrt{V_u}dW_{1,u}.$$
Now consider an affine process with $k=2$, and state space $\R \times \R_+$, from the above definitions it's clear that $X = (\log S_1, V )$ is indeed an affine Volterra process with diagonal kernel $\text{diag}(1,K)$ and coefficients, $\sigma(\cdot)$ in (\ref{affinevolterraprocess}) and $b(\cdot)$, $a(\cdot)$  in (\ref{e-affine-maps}) given by
\begin{align*}
 b^0 &= \begin{bmatrix} r_u \\ \kappa \phi \end{bmatrix}, \quad B=\begin{bmatrix} 0 & \theta - \frac{1}{2} \\ 0& -\kappa \end{bmatrix},\\
 \sigma(x) &= \begin{bmatrix} \sqrt{x} & 0 \\ \rho \sigma \sqrt{x} & \sqrt{1-\rho^2}\sigma \sqrt{x} \end{bmatrix}.
 \intertext{Then from (\ref{easigdecomposition}) we obtain}
  A^0 &= A^1 = 0, \quad A^2=\begin{bmatrix} 1 & \rho \sigma \\ \rho \sigma & \sigma^2 \end{bmatrix}.
\end{align*}
Also, Riccati-Volterra equation (\ref{ericcativolterra43}) takes the form
\begin{equation}
\begin{split}
\begin{bmatrix} \varphi_1 & \varphi_2 \end{bmatrix}&=\begin{bmatrix} \nu_1 & \nu_2 \end{bmatrix}\begin{bmatrix} 1 & 0 \\ 0& K \end{bmatrix}+\bigg(\begin{bmatrix} f_1 & f_2 \end{bmatrix}+\begin{bmatrix} \varphi_1 & \varphi_2 \end{bmatrix}\begin{bmatrix} 0 & \theta - \frac{1}{2} \\ 0& -\kappa \end{bmatrix}\\
&+\frac{1}{2}\begin{bmatrix} 0 & \begin{bmatrix} \varphi_1 & \varphi_2 \end{bmatrix} \begin{bmatrix} 1 & \rho \sigma \\ \rho \sigma & \sigma^2 \end{bmatrix}\begin{bmatrix} \varphi_1 \\ \varphi_2 \end{bmatrix}\end{bmatrix}\bigg)*\begin{bmatrix} 1 & 0 \\ 0& K \end{bmatrix},
\end{split}
\end{equation}
simplifying we get
\begin{align}
\varphi_1 & = \nu_1 + f_1, \label{evarphi1}\\
\varphi_2 & = \nu_2 K + \left(f_2 + (\theta -\frac{1}{2})\varphi_1-\kappa \varphi_2 +\frac{1}{2}(\varphi_1^2 + \sigma^2 \varphi_2^2 + 2\rho \sigma \varphi_1 \varphi_2)\right) \label{evarphi2}.
\end{align}

\noindent Let us now give an important assumption on the kernel $K$ which will be used throughout this chapter.

\begin{assumption}\cite[Assumption 2.5]{jaber_et_al_2017_affine_volterra_processes}\label{assump2}
	Assume, the kernel $K$ satisfies Assumption \ref{assump1} and is strictly positive and completely monotone, meaning
	$$(-1)^kK^{(k)}(t)\ge 0, \quad \fa t>0, \text{and $k \in \NN_+$}.$$
\end{assumption}

\begin{remark}
From Example 2.3 and 3.6 in \cite{jaber_et_al_2017_affine_volterra_processes} we have that the fractional kernel
$$K(t)=\frac{ct^{\alpha-1}}{\Gamma(\alpha)}$$
satisfies Assumption \ref{assump2}. Therefore, by letting $K$ in (\ref{evolterraheston}) be defined as the fractional kernel, we obtain the rough volatility model studied in \cite{euch_rosenbaum_2019_characteric_function_rough_heston_models}. 
\end{remark}

\noindent We now provide a crucial theorem in \cite{jaber_et_al_2017_affine_volterra_processes} which allows the use of Volterra Heston models to be used to define the financial market model.

\begin{theorem}\cite[Theorem 6.1]{jaber_et_al_2017_affine_volterra_processes}\label{t61}
	Assume $K$ satisfies assumption \ref{assump2}, then the stochastic Volterra equation
	 $$X_t = X_0 + \int_{0}^{t}K(t-u)b(X_u)du + \int_{0}^{t}K(t-u)\sigma \sqrt{X_u}dB_u,$$
	 has a unique in law $\R_+$-valued continuous weak solution $X$ for any initial condition $X_0\in \R_+$. The paths of $X$ are H\"{o}lder continuous of any order less than $\gamma/2$, where $\gamma$ is the constant associated with $K$ in assumption \ref{assump1}.	
\end{theorem}

\noindent From this theorem we obtain the main result of this section.

\begin{theorem}\cite[Theorem 7.1]{jaber_et_al_2017_affine_volterra_processes}
	Assume $K$ satisfies assumption \ref{assump2}. The stochastic Volterra equation (\ref{evolterraheston})-(\ref{estockprocess}) has a unique in law $\R \times \R_+$-valued continuous weak solution $(\log S, V)$ for any initial condition $(\log S_0,V_0)\in \R \times \R_+$. The paths of $V$ are H\"{o}lder continuous of any order less than $\gamma/2$, where $\gamma$ is the constant associated with $K$ in assumption \ref{assump1}.	
\end{theorem}
\begin{remark}
	This result follows directly from Theorem \ref{t61} along with the fact that $S$ is determined by $V$ (see Theorem 7.1 in \cite{jaber_et_al_2017_affine_volterra_processes} for further details).
\end{remark}

\noindent 
Let us now provide a number of useful results in relation to (\ref{evolterraheston}). These results will be relied upon in summarising the main results in \cite{han_wong_portfolio_utility_volterra_heston}.

\begin{proposition}\cite[Lemma 7.4]{jaber_et_al_2017_affine_volterra_processes}
	Assume $K$ satisfies assumption \ref{assump2}. Let $\nu \in (\C^2)^*$ and $f\in L_\text{loc}^1(\R_+;(\C^2)^*)$ be such that
		$$\textup{Re}(\varphi_1) \in [0,1], \, \textup{Re}(u_2) \le 0 \textup{ and Re}(f_2) \le 0, $$
		where $\varphi_1$ is given by (\ref{evarphi1}). Then the Riccati-Volterra equation (\ref{evarphi2}), with $\theta=0$ has a unique global solution $\varphi_2 \in L_\text{loc}^2(\R_+;\C^*)$, which satisfies $\textup{Re}(\varphi_2)\le 0$. 
\end{proposition}

\begin{proof}
	See Lemma 7.4 in \cite{jaber_et_al_2017_affine_volterra_processes}.
\end{proof}

\noindent Now, let us define a Riccati-Volterra equation which has the following form
\begin{equation}\label{ericcatisimple}
g(a,t)=\int_{0}^{t}K(t-u)[a-\kappa g(a,u)+\frac{\sigma^2}{2}g^2(a,u)]du, \quad g(a,0)=0.
\end{equation}

\noindent Han and Wong (see \cite{han_wong_portfolio_mean-variance_volterra_heston}) have obtained a number of useful existence and uniqueness results for Riccati-Volterra equations of this form (\ref{ericcatisimple}), which are used to obtain the optimal investment strategy in \cite{han_wong_portfolio_utility_volterra_heston}.

\begin{proposition}\cite[Lemma A.2]{han_wong_portfolio_mean-variance_volterra_heston}\label{phanricattiexstienceglobal}
	If $k^2-2a\sigma^2 > 0$, then (\ref{ericcatisimple}) has a unique global solution.
\end{proposition}
\begin{proof}
	See Lemma A.2. in \cite{han_wong_portfolio_mean-variance_volterra_heston}
\end{proof}

\noindent The following theorem immediately follows from the above result (see \cite{han_wong_portfolio_utility_volterra_heston}).

\begin{theorem}\cite[Theorem 2.5]{han_wong_portfolio_mean-variance_volterra_heston}\label{thanwong25}
	Suppose the Riccati-Volterra equation (\ref{ericcatisimple}) has a unique continuous solution on $[0,T]$ and $V$ is given by (\ref{evolterraheston}). Then
	\begin{equation}
	\begin{split}
	\E\left[\exp\left(a\int_{0}^{T} V_udu \right)\right]=&\exp \bigg[V_0\int_{0}^{T}\bigg(a-\kappa g(a,u)+\frac{\sigma^2}{2}g^2(a,u)\bigg)du \\
	&+\kappa \phi \int_{0}^{T}g(a,u)du\bigg]<\infty.
	\end{split}	
	\end{equation}
\end{theorem}

\noindent We now provide a final important Ricatti-Volterra existence and uniqueness result from \cite{han_wong_portfolio_mean-variance_volterra_heston}, as this result will be relied upon in \cite{han_wong_portfolio_utility_volterra_heston} to prove existence and uniqueness of (\ref{ericcatisol}).

\begin{theorem}\cite[Theorem A.1]{han_wong_portfolio_mean-variance_volterra_heston}\label{thanriccatiexistencelocal}
	Assume $K$ satisfies assumption \ref{assump2}. Let $c_0, \, c_1, \, c_2$ be constant. Then there exists $\delta >0$, such that
	\begin{equation}
	f(t)=(K*(c_0 + c_1f+c_2f^2))(t),\quad f(0)=0,
	\end{equation}
	has a unique continuous solution $f$ on $[0,\delta]$.
\end{theorem}

\noindent Let us now give an important martingale result which is used to obtain the exponential-affine transform representation (\ref{e434}) from Theorem \ref{t43}.
\begin{proposition}\label{pjaberexpmartingale}Let $g\in L^\infty(\R_+;\R)$ and define $U_t=\int_{0}^{t}g(u)\sqrt{V_u}dB_u, \, t\in [0,T]$, with $V$ given by (\ref{evolterraheston}) and $B$ given by (\ref{eBt}). Then the stochastic exponential $\exp(U_t - \frac{1}{2}\langle U \rangle_t)$ is a martingale.
	
\end{proposition}
The proof is in the same spirit as Lemma 7.3 in \cite{jaber_et_al_2017_affine_volterra_processes}.
\begin{proof} Let $M=\lbrace M_t,\F_t \sep 0\le t \le T \rbrace$ with $M_t \de \exp(U_t - \frac{1}{2}\langle U \rangle_t)$, and define the stopping times $\tau_n \de \inf \lbrace t\ge 0 \sep V_t > n\rbrace \wedge T$, thus $\P(\underset{n\to \infty}{\lim}\tau_n=T)=1$. Then $M_{t \wedge \tau_n}, \, t\in [0,T]$ is a uniformly integrable martingale for each $n$ by Novikov's condition, as 
	$$\E \left[\exp^{\frac{1}{2}\int_{0}^{T}g^2(t)V_{t \wedge \tau_n}dt} \right] < \infty.$$
	Moreover, as $M$ is a local martingale it is also a supermartingale by Fatou's lemma, as 
	$$\E[M_t|\F_u]=\E[\underset{n\to \infty}{\lim \inf}M_{t \wedge \tau_n}|\F_u]\le \underset{n\to \infty}{\lim \inf}\, \E[M_{t \wedge \tau_n}|\F_u]=\underset{n\to \infty}{\lim \inf}\, M_{u \wedge \tau_n}=M_u.$$
	Therefore, $\E[M_0]\le\E[M_T]$, however as $\E[M_0]=1$, to prove that $M$ is a martingale it suffices to show that $\E[M_T]\ge1$. Let's now define the probability measures $\Q^n$ by
	$$\frac{d\Q^n}{d\P}=M_{\tau_n} \quad n\ge 0.$$
	From Girsanov's Theorem, as $B$ is a Brownian motion with respect to $\P$, the process $dB_t^n=dB_t - d\langle U_{T\wedge \tau_n}, B \rangle_t = dB_t - \textbf{1}_{\lbrace t\le \tau_n \rbrace}g(t)\sqrt{V_t}dt$ is a Brownian motion under $\Q^n$, and we have
	$$V=V_0 + K* ((\kappa \phi-(\kappa-\sigma g\textbf{1}_{ [0,\tau_n] })V)dt + \sigma \sqrt{V}dB^n).$$
	Let $\gamma$ be the constant from Assumption \ref{assump1}, then choose $p$ sufficiently large such that $\gamma/2 - 1/p>0$. As $\kappa \phi-(\kappa-\sigma g\textbf{1}_{\lbrace t\le \tau_n(\omega) \rbrace})v$ satisfies a linear growth condition in $v$, uniformly in $(t,\omega)$, from Proposition \ref{plemma31} and Remark \ref{r32}, we have
	$$\underset{t\le T}{\sup}\E_{\Q^n}[|V_t|^p]\le c$$
	for some constant $c$ that depends only on $|X_0|,\, K|_{[0,T]},\, c_{LG}, \,p$ and $T$. In Proposition \ref{plemma24} we defined the $\alpha$-H\"{o}lder seminorm for any real-valued function $f$ as,
	$$|f|_{C^{0,\alpha}([0,T])} = \underset{0 \le s<t\le T}{\sup}\frac{|f(t)-f(s)|}{|t-s|^\alpha}.$$
	Using this substitution with $s=\alpha=0$ we get,
	\begin{align*}
	\Q^n(\tau_n < T)&=\Q^n(\underset{t<T}{\sup}V_t>n)\le \Q^n(\underset{t\le T}{\sup}V_t>n)\\
	&\le \Q^n(V_0 + |V|_{C^{0,0}([0,T])}>n).
	\intertext{Since,}
	\E_{\Q^n}[|V|^p_{C^{0,0}([0,T])}]&\ge \E_{\Q^n}[|V|^p_{C^{0,0}([0,T])}\textbf{1}_{\lbrace|V|_{C^{0,0}([0,T])}>(n-V_0\rbrace)}]\\
	 &\ge (n-V_0)^p\Q^n( |V|_{C^{0,0}([0,T])}>n-V_0),
	 \intertext{we have}
	\Q^n(\tau_n < T)&\le \left(\frac{1}{n-V_0}\right)^p\E_{\Q^n}[|V|^p_{C^{0,0}([0,T])}],
	\intertext{finally using proposition \ref{plemma24} we get}
	\Q^n(\tau_n < T) &\le \left(\frac{1}{n-V_0}\right)^pc',
	\end{align*}
	where $c'$ is a constant that only depends on $p, K$ and $T$. Therefore, by change of measure we have
	
	\begin{align*}
	\E_\P[M_T]&\ge \E_\P[M_T\textbf{1}_{\lbrace \tau_n=T\rbrace}]=\E_\Q[\textbf{1}_{\lbrace \tau_n=T\rbrace}],\\
	&=\Q^n(\tau_n=T)\le \Q^n(\tau_n=T)=1-\Q^n(\tau_n<T),\\
	&\ge 1-\left(\frac{1}{n-V_0}\right)^pc',
	\end{align*}
	by then sending $n$ to infinity yields $\E_\P[M_T]\ge1$, as required.
\end{proof}

\noindent We conclude this section with the following result which will be used to derive the Ansatz in the martingale optimality principle (see Definition \ref{dmartingaleoptimalityp}).

\begin{proposition}\cite[Proposition 3.2]{Keller-Ressel_et_al_2018_affine_rough_models}\label{pforwardvariance}
Let $R_\kappa$ be the resolvent of the second kind of $\kappa K$ (see Definition \ref{dres2}). The forward variance $\xi_u(t)=\E[V_s|\F_t]$, $0\le u\le t \le s \le T$ associated to (\ref{evolterraheston}) satisfies
$$d\xi_u(s)=\frac{1}{\kappa}R_\kappa (s-u)\sigma\sqrt{V_u}dW_u,$$
or equivalently
$$\xi_t(s)=\E[V_s|\F_t]=\xi_0(s) + \int_{0}^{t}\frac{1}{\kappa}R_\kappa (s-u)\sigma\sqrt{V_u}dW_u,$$
where
$$\xi_0(s)=V_0\left(1-\int_{0}^{s}R_\kappa (u)du\right)+\phi\int_{0}^{s}R_\kappa(u)du.$$
\end{proposition}
\begin{proof}
	See Proposition 3.2 in \cite{Keller-Ressel_et_al_2018_affine_rough_models}.
\end{proof}

\noindent Now that we've established the main results needed to solve the optimisation problem under the Volterra Heston model, let's now introduce the martingale distortion transformation. We will have a particular focus on how this method can be used to obtain the optimal investment strategy for Volterra Heston models.


\section{Solving the optimisation problem under the Volterra Heston model}\label{s_martingale_distortion_transformation}
The main goal of this section is to provide the motivation behind the construction of the Ansatz for the martingale optimality principle in \cite{han_wong_portfolio_utility_volterra_heston}. The Ansatz for this problem may originally seem quite a daunting process to construct. However, we will soon see that from historical works (see \cite{tehranchi_2004_martingale_distortion_non_markovian} and \cite{zariphopoulou_2001_supplement_martingale_distortion_markovian}), the form of the Ansatz is still heavily motivated from these results regardless of the choice of model dynamics. In addition, we summarise the main verification results of the Ansatz in \cite{han_wong_portfolio_utility_volterra_heston}. This work opens up future research considerations for Volterra Heston models under a number of different utility functions (i.e. exponential and logarithmic utility functions). 

\subsection{Solution to the optimisation problem under the Volterra Heston model}
Let us begin by first summarising the main results in \cite{han_wong_portfolio_utility_volterra_heston}. Han and Wong's approach to solving the optimisation problem (\ref{eoptimisation}) under the non-Markovian Volterra Heston model was to construct a family of processes, which satisfied the following definition.

\begin{definition}[Martingale optimality principle]\cite[Section 6.6.1, Section 3]{textbook_pham, han_wong_portfolio_utility_volterra_heston}\label{dmartingaleoptimalityp} The martingale optimality principle states that the optimisation problem (\ref{eoptimisation}) can be solved by constructing a family of processes $J^\pi = \lbrace J_t^\pi, \F_t | 0\le t \le T \rbrace, \, \pi \in \A$, satisfying the properties:
	\begin{enumerate}[label=(\roman*)]
		\item $J_T^\pi=U(\Pi^\pi_T), \fa \pi \in \A$
		\item $J_0^\pi$ is a constant, independent of $\pi \in \A$
		\item $J^\pi$ is a supermartingale $\fa \pi \in \A$, and there exists $\pi^*\in \A$ such that $J^\pi$ is a martingale.
	\end{enumerate}	
\end{definition}

\noindent Let us assume that the process $J^\pi$ satisfies the above conditions, with $U(\cdot)$ given by (\ref{eutility}) and $\Pi^*$ denotes the wealth process with respect to the optimal investment strategy $\pi^*$. Then, $\fa \pi \in \A$ we have
$$\E[U(\Pi^\pi_T)]=\E[J_T^\pi]\le J_0^\pi=J_0^{\pi^*}=\E[J_T^{\pi^*}]=\E[U(\Pi_T^*)].$$

\noindent From the above result we get $\E[U(\Pi_T^*)]\ge \E[U(\Pi^\pi_T)], \fa \pi \in \A$, and $J^{\pi^*}$ is a martingale process. Therefore, from (\ref{eoptimisationlocal}) we have that $J^{\pi^*}$ satisfies
$$J_t^{\pi^*}= \E\left[\frac{(\Pi_T^*)^\gamma}{\gamma}\bigg|\F_t\right]=\underset{\pi\in \A}{\sup}\,\E\left[\frac{(\Pi^\pi_T)^\gamma}{\gamma}\bigg|\F_t\right],$$  

\noindent as required. Therefore, constructing such a process becomes the main problem, which is the focus of this section. In \cite{han_wong_portfolio_utility_volterra_heston}, Han and Wong provide the Ansatz for $J^\pi$, however, the motivation behind the construction of this process is not supplied. We seek to close this gap and provide insight for future optimisation problems under the Volterra Heston model. 

Now, let us consider the Ansatz for $J^\pi$ in \cite{han_wong_portfolio_utility_volterra_heston}, which is given by

\begin{equation}\label{eansatzhan}
	J_t^\pi = \frac{(\Pi^\pi_t)^\gamma}{\gamma}M_t,
\end{equation}

\noindent where,
\begin{equation}
M_t = \exp \left\lbrace \int_{t}^{T}\left(\gamma r_u + \frac{\gamma \theta^2 \widetilde{\xi}_t(u)}{2(1-\gamma)}+\frac{\delta\sigma^2\widetilde{\xi}_t(u)}{2}\varphi^2(T-u)\right)du\right\rbrace,
\end{equation}

\noindent and $\delta=\frac{1-\gamma}{1-\gamma + \gamma \rho^2}$, and $\varphi(\cdot)$ satisfies the Riccati-Volterra equation
\begin{equation}
\varphi(t) = \int_{0}^{t}K(t-u) \left(\frac{\gamma \theta^2}{2\delta(1-\gamma)} - \lambda \varphi(u)  + \frac{\sigma^2}{2}\varphi^2(u)\right)du.
\end{equation}

\noindent The forward variance $\widetilde{\xi}_t(u)=\widetilde{\E}[V_u|\F_t]$ (see Proposition \ref{pforwardvariance}), is defined with respect to probability measure $\widetilde{\P}$, given by
\begin{equation}
\frac{d\widetilde{\P}}{d\P}=\exp\left(\frac{\gamma\theta}{1-\gamma}\int_{0}^{T}\sqrt{V_u}dW_u - \frac{\gamma^2\theta^2}{2(1-\gamma)^2}\int_{0}^{T}V_udu\right)	
\end{equation}

\noindent where $V$ is given by (\ref{evolterraheston}). Thus, re-writing $V$ with respect to $\widetilde{\P}$ we have,
\begin{equation}
V_t = v_0 + \int_{0}^{t}K(t-u)(\kappa \phi-\lambda V_u)du + \int_{0}^{t}K(t-u)\sigma \sqrt{V_u}d\widetilde{B}_u,
\end{equation}

\noindent where $\lambda = \kappa - \frac{\gamma}{1-\gamma}\rho \theta \sigma$. Now, in order to verify that (\ref{eansatzhan}) satisfies Definition \ref{dmartingaleoptimalityp}, we must verify that the process is indeed a supermartingale. To achieve this, we require that the process $M$ must be integrable, and establish the dynamics of $M$. Han and Wong achieve this result by imposing certain conditions, summarised in the following theorem.

\begin{theorem}\cite[Theorem 3.1]{han_wong_portfolio_utility_volterra_heston}\label{than31} Assume
	\begin{equation}\label{eansatzconditions}
	\kappa^2 - 6\frac{\gamma^2}{(1-\gamma)^2}\theta^2 \sigma^2 >0, \quad \lambda>0,\quad \lambda^2-2p\frac{\gamma}{1-\gamma}\theta^2\sigma^2 >0,
	\end{equation}	
	for some $p>1/(2\delta)$. Then $M$ has the following properties:
	\begin{enumerate}[label=(\roman*)]
		\item $M_t\ge c > 0$ for some positive constant $c$, and $\E[\underset{t\in [0,T]}{\sup}|M_t|^p]<\infty$,
		\item \begin{equation}
		\begin{split}
		dM_t = &-\left(\gamma r_t + \frac{\gamma\theta^2V_t}{2(1-\gamma)}\right)M_tdt - \frac{\gamma}{2(1-\gamma)}\left(2\theta \sqrt{V_t}U_{1t}+\frac{U_{1t}^2}{M_t}\right)dt\\
		&+ U_{1t}dW_{1t}+U_{2t}dW_{2t},
		\end{split}
		\end{equation}
		where
		\begin{equation}
		U_{1t}=\rho \delta \sigma M_t \sqrt{V_t}\varphi(T-t),
		\end{equation}
		\begin{equation}
		U_{2t}=\sqrt{1-\rho^2} \delta \sigma M_t \sqrt{V_t}\varphi(T-t).
		\end{equation}
		\item $\E\left[\left(\int_{0}^{T}U^2_{iu}du\right)^{p/4}\right]<\infty,$ for $i=1,2.$
	\end{enumerate}	
\end{theorem}

\begin{remark}
	The proof in part (i) is dependent on the assumption of deterministic interest rates, and utilises H\"{o}lder's and Doob's maximal inequalities (see \cite{textbook_karatzas_shreve_stochastic_calcuus}). Moreover, the finite bounds are established from the direct use Theorem \ref{thanwong25}, which requires Proposition \ref{phanricattiexstienceglobal} to be satisfied. Thus, establishing the dependence on conditions (\ref{eansatzconditions}) in the proof. Part (ii) utilises It\^o's lemma, as well as Theorem \ref{theoremstochasticfubini}, and part (iii) is an immediate result due to part (i).
\end{remark}

Using the results of Theorem \ref{than31} Han and Wong then proceed to verify that the Ansatz (\ref{eansatzhan}) for $J^\pi$ does indeed satisfy Definition \ref{dmartingaleoptimalityp}. Thus, there exists an optimal investment strategy $\pi^* \in \A$ such that $J^{\pi^*}$ is a martingale process.

\begin{theorem}\cite[Theorem 3.2]{han_wong_portfolio_utility_volterra_heston}\label{than32} Suppose the conditions (\ref{eansatzconditions}) in Theorem \ref{than31} hold and $\kappa^2-2\eta\sigma^2>0$, where
	$$\eta = \max \lbrace 2a |\theta| \underset{t\in [0,T]}{\sup}|A_t|, 2a(4a-1)\underset{t\in [0,T]}{\sup}|A_t|^2\rbrace,$$ 
	for some $a>1$, and $A_t= \frac{1}{1-\gamma}[\theta + \rho \delta \sigma \varphi(T-t)]$. Then the process
	\begin{equation}
	J_t^\pi = \frac{(\Pi^\pi_t)^\gamma}{\gamma}M_t, \quad t\in [0,T]
	\end{equation}
	with $(M)_{t\in [0,T]}$ given by (\ref{emtsol}) satisfies the martingale optimality principle (see Definition \ref{dmartingaleoptimalityp}), and the optimal strategy is given by
	\begin{equation}\label{epiopt}
	\pi^* = A_t
	\end{equation}
	Moreover,
	\begin{equation}
	\E[\underset{t\in [0,T]}{\sup}|\Pi_t^*|^q]<\infty,
	\end{equation}
	and $\pi^*$ is admissible.
\end{theorem}

\begin{remark}
	This result follows directly from the semimartingale decomposition of $J_t^\pi$, which is given by
	$$dJ_t^\pi = J_t^\pi F(\pi,t)dt + J_t^\pi \left(\frac{U_{1t}}{M_t}+\gamma\pi \sqrt{V_t}\right)dW_{1t}+J_t^\pi\frac{U_{2t}}{M_t}dW_{2t}$$
	where 
	$$F(\pi,t)=\pi^2\frac{\gamma(\gamma-1)}{2}V_t + \pi\left(\theta \sqrt{V_t}+\frac{U_{1t}}{M_t}\right) \gamma \sqrt{V_t} - \frac{\gamma}{2(1-\gamma)}\left(\theta \sqrt{V_t}+\frac{U_{1t}}{M_t}\right)^2,$$
	with (\ref{epiopt}) derived from $\frac{\partial F}{\partial \pi}=0$. Moreover, as $F(\pi,t)$ is quadratic in $\pi$ with negative leading coefficient ($\gamma-1<0$), and $F(\pi^*,t)=0 \implies F(\pi,t)\le0, \fa \pi\in \A, \, t\in [0,T]$. Therefore, the process $J^\pi$ is a semimartingale. Finally, the process 
	$$dJ_{t}^{\pi^*} = J_t^{\pi^*} \left(\frac{U_{1t}}{M_t}+\gamma\pi^* \sqrt{V_t}\right)dW_{1t}+J_t^{\pi^*}\frac{U_{2t}}{M_t}dW_{2t},$$
	is a martingale by Proposition \ref{pjaberexpmartingale}, and since $J_0^\pi$ is independent of $\pi$, and $J_T^\pi=\Pi^\pi_T$, we have that the process $J^\pi$ does indeed satisfy the martingale optimality principle as required.
\end{remark}

Now that we've established the solution provided by Han and Wong to the optimisation problem under a Volterra Heston model. Let's now provide the motivation required to construct the process (\ref{eansatzhan}).
\subsection{Construction of the Ansatz using the martingale distortion transformation}\label{sconstructionansatz}

The martingale distortion transformation is motivated by Proposition \ref{pdistortiontransform}, specifically equations (\ref{edistortion}), (\ref{edistortionsolution}) and (\ref{edistortionbrownian}). Tehranchi in 2004 (see \cite{tehranchi_2004_martingale_distortion_non_markovian}), recognised that the solution to the Merton portfolio optimisation problem was consistently displaying a similar form mentioned in Section \ref{soptimisationsolution}. His proposed method, the martingale distortion transformation, circumvented the dependency of Markovian structure, thus providing a method which also accommodates non-Markovian models. Let us now reformulate his proposed methodology under the financial model defined in Section \ref{sfinancialmarketandoptimisationproblem} with the help of \cite{foque_hu_2019_optimal_portfolio_fractional_martingale_distortion}.

From (\ref{ewealthprocess}) the Sharpe-ratio is given as 
\begin{equation}\label{esharpe}
\lambda(y)=\theta \sqrt{y},
\end{equation}
and is essentially bounded by Proposition \ref{plemma31}. We must now define a new probability measure $\widetilde{\P} \sim \P$, such that (\ref{edistortionbrownian}) is satisfied. With this condition in mind, let's define the Radon-Nikod\'ym density of $\widetilde{\P}$ with respect to $\P$ as
\begin{equation}\label{eprobmeasure}
\frac{d\widetilde{\P}}{d\P}=\exp\left(\frac{\gamma}{1-\gamma}\int_{0}^{T}\lambda(V_u)dW_u - \frac{\gamma^2}{2(1-\gamma)^2}\int_{0}^{T}\lambda^2(V_u)du\right).	
\end{equation}
Checking that this does indeed satisfy (\ref{edistortionbrownian}), we have from Girsanov's Theorem that the process $\widetilde{B}$ is given by
\begin{align}
d\widetilde{B}_u&=dB_t -\frac{\rho \gamma}{1-\gamma}\lambda(V_u)du,\\
&=\rho \left(dW_{1u}-\frac{\gamma \theta}{1-\gamma}\sqrt{V_u}du\right)+\sqrt{1-\rho^2}dW_{2u},\\
&=\rho d\widetilde{W}_{1u}+\sqrt{1-\rho^2}dW_{2u},
\end{align}
and is a Brownian motion under $\widetilde{\P}$ as required. Then, motivated by (\ref{edistortion}) and (\ref{edistortionsolution}), and using the results in \cite{foque_hu_2019_optimal_portfolio_fractional_martingale_distortion} and \cite{tehranchi_2004_martingale_distortion_non_markovian}, we arrive at the following proposition after substitution of (\ref{esharpe}).

\begin{proposition}\cite[Proposition 2.2, Theorem 3.1-3.2]{foque_hu_2019_optimal_portfolio_fractional_martingale_distortion, han_wong_portfolio_utility_volterra_heston}
	Under Assumption \ref{assump2} suppose conditions (\ref{eansatzconditions}) are satisfied, then the value process is given by
	\begin{equation}\label{eoptimal}
	\J_t^\pi = \frac{(\Pi^\pi_t)^\gamma}{\gamma}\left(\widetilde{\E}\left[\exp\left\lbrace \int_{t}^{T}\frac{\gamma}{\delta}\left(r_u +\frac{\theta^2V_u}{2(1-\gamma)}\right)du\right\rbrace\bigg| \F_t\right]\right)^\delta,
	\end{equation}
	where $\delta$ is given by (\ref{edistortionpower}) and optimal investment strategy $\pi^*$ is given by (\ref{epiopt}).
\end{proposition}

\begin{proof}
 Let us begin by writing the variance process (\ref{evolterraheston}) under the probability measure $\widetilde{\P}$ (\ref{eprobmeasure}),
\begin{equation}\label{evolterraheston2}
V_t = v_0 + \int_{0}^{t}K(t-u)(\kappa \phi-\lambda V_u)du + \int_{0}^{t}K(t-u)\sigma \sqrt{V_u}d\widetilde{B}_u,
\end{equation}

\noindent where $\lambda = \kappa - \frac{\gamma}{1-\gamma}\rho \theta \sigma$. Then from (\ref{eoptimal}), $J_0^\pi$ is independent of $\pi$ and takes the form
\begin{equation}\label{ehan_j0}
J_0^\pi = \frac{w_0^\gamma}{\gamma}\left(\widetilde{\E}\left[\exp\left\lbrace \int_{0}^{T}\frac{\gamma}{\delta}\left(r_u +\frac{\theta^2V_u}{2(1-\gamma)}\right)du\right\rbrace\right]\right)^\delta.
\end{equation}

\noindent Define the process $M=\lbrace M_t, \F_t | 0\le t \le T\rbrace$, with 
\begin{equation}\label{emt1}
	M_t^{1/\delta} = \widetilde{\E}\left[\exp\left\lbrace \int_{t}^{T}\frac{\gamma}{\delta}\left(r_u +\frac{\theta^2V_u}{2(1-\gamma)}\right)du\right\rbrace\bigg| \F_t \right].
\end{equation}

\noindent Now, let $f=\frac{\gamma \theta^2}{2\delta(1-\gamma)}$, $X$ be given by (\ref{evolterraheston2}), and $\nu=0$. Then the Riccati-Volterra equation (\ref{ericcativolterra43}) takes the form,

\begin{equation}\label{ericcatisol}
\varphi = K* \left(\frac{\gamma \theta^2}{2\delta(1-\gamma)} - \lambda \varphi  + \frac{\sigma^2}{2}\varphi^2\right),
\end{equation}

\noindent and from Theorem \ref{thanriccatiexistencelocal}, existence and uniqueness for (\ref{ericcatisol}) is established. Moreover, if $\lambda>0$ and $\lambda^2 -\frac{\gamma\theta^2\sigma^2}{\delta(1-\gamma)}>0$, then (\ref{ericcatisol}) has a unique nonnegative global solution by Proposition \ref{phanricattiexstienceglobal}. The process (\ref{e431}) is then given by
\begin{equation}
Y_t = Y_0 + \int_{0}^{t}\varphi(T-u)\sigma \sqrt{V_u}d\widetilde{B}_u - \frac{1}{2}\int_{0}^{t}\varphi(T-u)\sigma^2V_u\varphi(T-u)^Tdu,
\end{equation}
with
\begin{equation}
Y_0 = \int_{0}^{T}\left(\frac{\gamma \theta^2}{2\delta(1-\gamma)}V_0 -\lambda \varphi(u)V_0+\frac{1}{2}\varphi^2(u)\sigma^2V_0\right)du.
\end{equation}
As (\ref{ericcatisol}) is essentially bounded we have by Proposition \ref{pjaberexpmartingale} that $\exp(Y)$ is a martingale. Thus, we obtain the following relationship from (\ref{e433}) and the exponential-affine transform representation (\ref{e434})
\begin{equation*}
\begin{split}
\widetilde{\E}\left[\exp \left\lbrace\int_{0}^{T}\frac{\gamma \theta^2}{2\delta(1-\gamma)}V_udu \right\rbrace \bigg |\F_t\right]=&\exp \bigg\lbrace \widetilde{\E}\left[\int_{0}^{T}\frac{\gamma \theta^2}{2\delta(1-\gamma)}V_udu \bigg | \F_t \right]\\
 &+ \frac{1}{2}\int_{t}^{T}\varphi(T-u)^2\widetilde{\E}[V_u|\F_t]du\bigg \rbrace. 
\end{split}
\end{equation*}

\noindent Then, using Proposition \ref{pforwardvariance} with respect to the process (\ref{evolterraheston2}) and simplifying, we get
$$\widetilde{\E}\left[\exp \left\lbrace\int_{t}^{T}\frac{\gamma \theta^2V_u}{2\delta(1-\gamma)}du \right\rbrace \bigg |\F_t\right]=\exp \bigg\lbrace \int_{t}^{T}\left(\frac{\gamma \theta^2\widetilde{\xi}_t(u)}{2\delta(1-\gamma)} + \frac{\sigma^2\widetilde{\xi}_t(u)}{2}\varphi^2(T-u) \right)du.$$

\noindent The above result combined with (\ref{emt1}) then yields
\begin{equation}\label{emtsol}
M_t = \exp \left\lbrace \int_{t}^{T}\left(\gamma r_u + \frac{\gamma \theta^2 \widetilde{\xi}_t(u)}{2(1-\gamma)}+\frac{\delta\sigma^2\widetilde{\xi}_t(u)}{2}\varphi^2(T-u)\right)du\right\rbrace.
\end{equation}

\noindent Re-writing (\ref{eoptimal}) we have
\begin{equation}\label{eansatzsimple}
J_t^\pi = \frac{(\Pi^\pi_t)^\gamma}{\gamma}M_t, \quad t\in [0,T].
\end{equation}

\noindent By then comparing (\ref{eansatzsimple}) with the Ansatz in \cite{han_wong_portfolio_utility_volterra_heston} (\ref{eansatzhan}), we can see that we've arrived at the same result. Then, from Theorem's \ref{than31}-\ref{than32}, we can also see that the optimal investment strategy $\pi^*$, is indeed given by (\ref{epiopt}).
\end{proof}

Therefore, using the martingale distortion transformation developed by Tehranchi (see \cite{tehranchi_2004_martingale_distortion_non_markovian}), we have a means of constructing a good initial hypothesis for the Ansatz $J^\pi$, which satisfies Definition \ref{dmartingaleoptimalityp} under the Volterra Heston model. From this approach we've shed light on the motivation behind the construction of the Ansatz in \cite{han_wong_portfolio_utility_volterra_heston} which opens up future research opportunities to consider additional, and more general, utility functions.


\chapter{A finite dimensional approach to solving the portfolio optimisation problem under rough Heston models}\label{mbd}
In this chapter we consider a finite dimensional approximation approach to solving the optimisation problem (\ref{eoptimisation}) given by \bd in \cite{bauerle_desmettre_2019_portfolio_optimization_des_baulere}. The key difference between this approach and the martingale distortion approach given in Section \ref{s_martingale_distortion_transformation}, is that it casts the non-Markovian optimisation problem into the classical optimisation framework, which utilises the Hamilton-Jacobi-Bellman (HJB) equation. We begin the chapter by introducing the rough Heston model used in this approach, and then introduce the finite dimensional approach. Finally, we conclude the chapter with a formal derivation of the HJB equation which naturally leads to the solution of the optimisation problem (\ref{eoptimisation}) under this approach.

\section{A rough Heston model via the Marchaud fractional derivative}

Under the financial market model defined in Section \ref{sfinancialmarketandoptimisationproblem}, with filtration $\FF = \lbrace \F_t \rbrace_{0 \le t \le T}$ given by the augmented filtration with respect to $W$ (\ref{e-fil-aug-bm}), B\"auerle and Desmettre consider a rough volatility model with Hurst index $H\in (0,1/4)$ (see \cite{bauerle_desmettre_2019_portfolio_optimization_des_baulere}). They achieve this by considering the \textit{Marchaud fractional derivative} (see \cite{textbook_fractional_integrals_derivatives})
\begin{equation}
D_0^{\alpha+1}f(t)=f(t)\frac{t^{-\alpha-1}}{\Gamma(-\alpha)}+ \frac{\alpha+1}{\Gamma(-\alpha)}\int_{0}^{t}\frac{f(t)-f(u)}{(t-u)^{\alpha+2}}du,
\end{equation}

\noindent with $\alpha = 2H -1 \in (-1,1/2)$. The main reason for B\"auerle and Desmettre to consider the Marchaud fractional derivative over the Riemann-Liouville fractional derivative used in \cite{guennoun_et_al_2018_asymptotic_behavior_fractional_heston}, is that the Marchaud fractional derivative is defined for \h $\delta$-continuous functions, with $\alpha+1<\delta \le 1$. Whereas, the Riemann-Liouville fractional derivative
requires $f$ to be absolutely continuous, which is not possible to ensure by direct consequence of the following result.
\begin{proposition}\cite[Lemma 7.1]{bauerle_desmettre_2019_portfolio_optimization_des_baulere}\label{pb_zholder}
	The paths of (\ref{eb_zdynamics}) are almost H\"{o}lder continuous on $[0,T]$ of any order $\delta<1/2$.
\end{proposition}

\begin{proof}
	See Lemma 7.1 in \cite{bauerle_desmettre_2019_portfolio_optimization_des_baulere}.
\end{proof}

\noindent Moreover, as $\delta<1/2$ (see Proposition \ref{pb_zholder})), $\alpha$ has to be from the interval $(-1,-1/2) \implies H\in (0,1/4)$. Now, to define the variance process $V$ in (\ref{estockprocess}) we begin by first introducing the volatility process $\nu = \lbrace \nu_t,\F_t \sep 0\le t \le T\rbrace$,

\begin{equation}\label{eb_rough_volatility}
\nu_t = \nu_0  + Z_t \frac{t^{-\alpha-1}}{\Gamma(-\alpha)}+ \frac{\alpha +1}{\Gamma(-\alpha)}\int_{0}^{t}\frac{Z_t-Z_u}{(t-u)^{\alpha+2}}du,
\end{equation}
where
\begin{equation}\label{eb_zdynamics}
dZ_t = \kappa(\phi - Z_t)dt + \sigma\sqrt{Z_t}dB_t,
\end{equation}

\noindent with $B$ defined by (\ref{eBt}), and parameters $(\nu_0, \kappa, \phi, \sigma, Z_0)\in \R_+$. The paths of the volatility process (\ref{eb_rough_volatility}) exhibit a rough behaviour, which is illustrated in \cite{bauerle_desmettre_2019_portfolio_optimization_des_baulere}. Moreover,
$$\underset{\alpha\downarrow -1}{\lim}D_0^{\alpha+1}f(t)=f(t),$$
meaning, that in the limiting case (i.e. $\alpha\downarrow -1$), the classical Heston model is obtained.

 Now, let $y=(t-u)x$, with $0\le u \le t \le T$, and as $\alpha \in (-1,-1/2)$ we have
\begin{align*}
(\alpha+1)\Gamma(\alpha+1)&=\int_{0}^{\infty}y^{\alpha+1}e^{-y}dy,
\intertext{by change of variables we get}
(\alpha+1)\Gamma(\alpha+1) &=\int_{0}^{\infty}((t-u)x)^{\alpha+1}e^{-(t-u)x}(t-u)dx,
\end{align*}
thus,
\begin{equation}
(t-u)^{-\alpha-2}=\frac{\Gamma(-\alpha)}{\alpha+1}\int_{0}^{\infty}e^{-(t-u)x}\widetilde{\mu}(dx),\quad \text{with $\widetilde{\mu}(dx)=\frac{x^{\alpha+1}dx}{\Gamma(-\alpha)\Gamma(\alpha+1)}$}.
\end{equation}
Using this result along with Fubini's Theorem, we obtain
\begin{align}
\nu_t &= \nu_0  + Z_t \frac{t^{-\alpha-1}}{\Gamma(-\alpha)}+ \int_{0}^{t}(Z_t-Z_u) \left(\int_{0}^{\infty}e^{-(t-u)x}\widetilde{\mu}(dx)\right)du, \nonumber \\
 &= \nu_0  + Z_t \frac{t^{-\alpha-1}}{\Gamma(-\alpha)}+ \int_{0}^{\infty} \left(\int_{0}^{t}(Z_t-Z_u)e^{-(t-u)x}du\right)\widetilde{\mu}(dx), \nonumber \\
  &= \nu_0  + Z_t \frac{t^{-\alpha-1}}{\Gamma(-\alpha)}+ \int_{0}^{\infty} \widetilde{Y}_t^x\widetilde{\mu}(dx), \label{eb_volatility_process}
\end{align}
where
\begin{equation}\label{eb_y_solution}
\widetilde{Y}_t^x\de \int_{0}^{t}(Z_t-Z_u)e^{-(t-u)x}du.
\end{equation}

\noindent Now, re-writing $\widetilde{Y}_t^x$ as
\begin{align*}
\widetilde{Y}_t^x &= Z_te^{-tx}\int_{0}^{t}e^{ux}du - e^{-tx}\int_{0}^{t}Z_ue^{ux}du,\\
&= Z_t\frac{1}{x}(1-e^{-tx}) - e^{-tx}\int_{0}^{t}Z_ue^{ux}du,
\end{align*}

\noindent from It\^o's lemma and the Fundamental Theorem of Calculus, we can see that $\widetilde{Y}_t^x$ satisfies the stochastic differential equation
\begin{align}
d\widetilde{Y}_t^x&= dZ_t\left(\frac{1}{x}(1-e^{-tx})\right) + Z_te^{-tx}dt + xe^{-tx}dt\int_{0}^{t}Z_ue^{tu}du-e^{-tx}Z_te^{tx}dt, \nonumber \\ 
&=\left(\frac{1}{x}(1-e^{-tx})\kappa(\phi - Z_t)-x\widetilde{Y}_x^t\right)dt+\frac{1}{x}(1-e^{-tx})\sigma\sqrt{Z_t}dB_t. \label{eb_y_sde}
\end{align}

\noindent Therefore, $\widetilde{Y}_t^x$ is a diffusion process, and by extension a Markov process. Unfortunately, the volatility process $\P(\nu\ge0) <1$. To remedy this short-coming, B\"auerle and Desmettre consider a measurable function $a:\R\to \R_+$, which then allows them to define the variance process $V$, by $V_t \de a(\nu_t), \fa t\in [0,T]$. Recall that the stock price process $S_1$ defined by (\ref{estockprocess}), reads
$$dS_{1,t} = S_{1,t}(r_t + \theta V_t)dt + S_{1,t} \sqrt{V_t}dW_{1,t}. $$

\noindent Substituting $a(\nu_t)$ we then get,
\begin{equation}\label{eb_stock_process}
dS_{1,t} = S_{1,t}(r_t + \theta a(\nu_t))dt + S_{1,t} \sqrt{a(\nu_t)}dW_{1,t}.
\end{equation}

\noindent Therefore, the wealth process (\ref{ewealthprocess}) becomes
\begin{equation}\label{eb_wealth_process}
	d\Pi^\pi_t = \Pi^\pi_t(r_t + \theta \pi_ta(\nu_t)) dt + \Pi^\pi_t\pi_t \sqrt{a(\nu_t)}  dW_{1,t},
\end{equation}

\noindent where we assume that $\Pi^\pi_0 = w_0 >0$, is the given initial wealth. Now, in order to solve the optimisation problem (\ref{eoptimisation}) in the classical framework (see Proposition \ref{pdistortiontransform}), \bd consider a finite dimensional approximation approach which is inspired by the work in \cite{carmon_et_al_2000}, given by the next section.

\section{A finite dimensional approximation of the optimisation problem}

The main idea in the finite dimensional approximation used in \cite{bauerle_desmettre_2019_portfolio_optimization_des_baulere}, is to approximate $\mu$ in (\ref{eb_volatility_process}) by a discrete measure with a finite number of atoms. This procedure then casts the non-Markovian process $\nu$, into a finite sum of diffusion processes, and thus allowing $\nu$ to be represented as a diffusion process. This is achieved by a quantization of $\mu$. Let us begin by defining the partition 
\begin{equation}\label{eb_partition}
\ZZ^n \de \lbrace 0 < \xi_0^n < \dots < \xi_n^n < \infty \rbrace,
\end{equation}

\noindent and define the barycentre of $\mu$ on the respective intervals $(\xi_i^n,\xi_{i+1}^n)$ for $i=0,\dots,n-1$ by 
\begin{align}
x_{i+1}^n &\de \frac{\int_{\xi_i^n}^{\xi_{i+1}^n}x\mu(dx)}{\int_{\xi_i^n}^{\xi_{i+1}^n}\mu(dx)} \label{eb_x_i},
\intertext{and the mass on the atom for $i=0,\dots,n-1$ by}
q_{i+1}^n &\de \int_{\xi_i^n}^{\xi_{i+1}^n}\mu(dx) \label{eb_q_i}.
\intertext{Then, the corresponding measure is given by}
\mu^n &\de \sum_{i=1}^{n}q_i^n\delta_{x_i^n} \label{eb_mu_n},
\end{align}
where $\delta_x$ is the Dirac measure on $x$. Now, in order ensure convergence, the following assumption on $\ZZ^n$ is required.

\begin{assumption}\label{ab_partition}Let $\ZZ^n$ be given by (\ref{eb_partition}) and assume
\begin{enumerate}[label=(\roman*)]
	\item $\xi_0^n\to 0$ and $\underset{n\to \infty}{\lim}\xi_n^n\to \infty$,
	\item $\delta(\ZZ^n)\de \underset{n\to \infty}{\lim}\max_{i=0,\dots,n-1}|\xi_{i+1}^n-\xi_i^n|\to 0$, and
	\item $\ZZ^n \subset \ZZ^{n+1}$.
\end{enumerate}
\end{assumption} 

\noindent From this assumption \bd obtain the following convergence results.

\begin{proposition}\cite[Lemma 3.1]{bauerle_desmettre_2019_portfolio_optimization_des_baulere}\label{pb_convergence}
	Suppose $f\in L^1(\mu)$ and the sequence $(\ZZ^n)_{n\ge1}$ given by (\ref{eb_partition}) satisfies Assumption \ref{ab_partition}. Then
	\begin{equation}\label{eb_convergence}
	\int fd\mu^n \to \int f d\mu, \quad n\to \infty,
	\end{equation}
	if $f\ge0$ and $f$ is convex. 
\end{proposition}
\begin{proof}
	See Section 7.3 in \cite{bauerle_desmettre_2019_portfolio_optimization_des_baulere}.
\end{proof}
\begin{remark}
	It is also important to note that due to the conditions in Assumption \ref{ab_partition}, the convergence of (\ref{eb_convergence}) is monotone increasing.
\end{remark}

\noindent Now, from (\ref{eb_x_i})-(\ref{eb_mu_n}), let us denote the finite dimensional approximate volatility process by
\begin{equation}\label{eb_volatility_approx}
\nu_t^n \de \nu_0  + Z_t \frac{t^{-\alpha-1}}{\Gamma(-\alpha)}+ \int_{0}^{\infty} \widetilde{Y}_t^x\widetilde{\mu}(dx)=\nu_0  + Z_t \frac{t^{-\alpha-1}}{\Gamma(-\alpha)}+ \sum_{i=1}^{n} \widetilde{Y}_t^{x_i^n}\widetilde{q}_i^n,
\end{equation}

\noindent  which leads naturally to the following convergence theorem in \cite{bauerle_desmettre_2019_portfolio_optimization_des_baulere}.

\begin{theorem}\cite[Theorem 3.2]{bauerle_desmettre_2019_portfolio_optimization_des_baulere}
	Let $a$ be a function of class $C^2(\R;\R_+)$ and convex. Then, under the assumptions of Proposition \ref{pb_convergence}
	\begin{equation}
	\underset{n\to \infty}{\lim}a(\nu_t^n)\uparrow a(\nu_t), \quad 0\le t < \infty,  \, \P\text{-a.s..}
	\end{equation}	
\end{theorem}
\begin{remark}
	This is a direct result of Proposition \ref{pb_convergence} (see Theorem 3.2 in \cite{bauerle_desmettre_2019_portfolio_optimization_des_baulere} for further details).
\end{remark}

\noindent Using the above results, the stock price process (\ref{eb_stock_process}) with respect to the approximate volatility process (\ref{eb_volatility_approx}), has dynamics given by 
\begin{equation}\label{eb_stock_process_approx}
dS_{1,t} = S_{1,t}(r_t + \theta a(\nu_t^n))dt + S_{1,t} \sqrt{a(\nu_t^n)}dW_{1,t}, \quad S_{1,0}=s_{1,0}>0.
\end{equation}

\noindent Similarly, the wealth process (\ref{eb_wealth_process}) with respect to the approximate volatility process (\ref{eb_volatility_approx}), evolves according to 
\begin{equation}\label{eb_wealth_process_approx}
d\Pi^\pi_t = \Pi^\pi_t(r_t + \theta \pi_ta(\nu_t^n)) dt + \Pi^\pi_t\pi_t \sqrt{a(\nu_t^n)}  dW_{1,t}, \quad \Pi^\pi_0 = w_0>0.
\end{equation}

\noindent Thus, the stochastic optimal control problem (\ref{eoptimisation}) is cast into the classical Markovian framework, and is defined by
\begin{align}\label{eb_optimsation}
\J_t^* &\de \underset{\pi\in \A}{\sup}\,\E \left[\frac{(\Pi^\pi_T)^\gamma}{\gamma} \bigg|\F_t\right], \nonumber \\
& = \underset{\pi\in \A}{\sup}\,\E \left[\frac{(\Pi^\pi_T)^\gamma}{\gamma} \bigg|\Pi^\pi_t=w, \, \widetilde{Y}_t^{x_1^n}=\widetilde{y}_1,\dots, \widetilde{Y}_t^{x_n^n}=\widetilde{y}_n,Z_t=z \right]. 
\end{align}

\noindent The solution to the finite dimensional classical stochastic optimal control problem (\ref{eb_optimsation}) is obtained from the use of the HJB equation. In order to understand the conditions imposed on the solution, it is prudent to give a brief introduction to the HJB equation.

\section{Solving the finite dimensional optimisation problem}

To understand the derivation of the HJB equation, let us begin by first introducing the concept of the Dynamic Programming Principle (DPP), which is a fundamental principle in the theory of stochastic control. Let
$$J^\pi(t,w, \widetilde{y}_1, \dots, \widetilde{y}_n, z) = \E[U(\Pi^\pi_T)|\Pi^\pi_t=w, \, \widetilde{Y}_t^{x_1^n}=\widetilde{y}_1,\dots, \widetilde{Y}_t^{x_n^n}=\widetilde{y}_n, \,Z_t=z],$$
be defined as the gain function, where $U$ is given by (\ref{eutility}), and $\pi \in \A$. Let $\tau \in \TT_{t,T}$ be the set of stopping times valued in $[t,T]$. Then by iterated conditioning, we have
\begin{align*}
J^\pi&(t,w, \widetilde{y}_1, \dots, \widetilde{y}_n, z) \\
&=\E[J^\pi(\tau,\Pi^\pi_\tau, \, \widetilde{Y}_\tau^{x_1^n},\dots, \widetilde{Y}_\tau^{x_n^n},Z_\tau)|\Pi^\pi_t=w, \, \widetilde{Y}_t^{x_1^n}=\widetilde{y}_1,\dots, \widetilde{Y}_t^{x_n^n}=\widetilde{y}_n,Z_t=z],
\end{align*}
with terminal condition $J^\pi(T,w, \widetilde{y}_1, \dots, \widetilde{y}_n, z) = \frac{w^\gamma}{\gamma}, \, \fa w\in \R_+, \,  \widetilde{y}_i \in \R, \, z \in \R_+$. This result leads naturally to the following fundamental theorem.
\begin{theorem}[Dynamic Programming Principle (DPP) - finite horizon]\cite[Theorem 3.3.1]{textbook_pham}\label{tdpp} For any stopping time $\tau \in \TT_{t,T}$, and diffusion process  
	\begin{equation}\label{ep_diffusion}
	dX_t = b(t, \pi_t, X_t)dt + \sigma(t, \pi_t, X_t)dW_t, \quad X_0=x_0 \in \R^k,
	\end{equation}
	which has a strong unique solution $X_t^{x_0}\in \R^k,\, \fa t\in[0,T]$, for some $k\in \NN_+$ (meaning the coefficients $b(\cdot)$ and $\sigma(\cdot)$ in (\ref{ep_diffusion}) satisfy Lipschitz continuity and the linear growth condition (see \cite{textbook_karatzas_shreve_stochastic_calcuus} or \cite{textbook_pham} for further details)), we have
	\begin{align}
	G(t,x) &= \underset{\pi \in \A}{\sup} \; \underset{\tau \in \TT_{t,T}}{\sup} \, \E\left[G(\tau,X_\tau)|X_t = x\right], \label{ep_dpp_1}\\
	 &= \underset{\pi \in \A}{\sup} \; \underset{\tau \in \TT_{t,T}}{\inf} \, \E\left[G(\tau,X_\tau)|X_t = x\right], \label{ep_dpp_2}
	\end{align}
	where $x\in \R^k$.
\end{theorem}

\begin{remark}
	From (\ref{ep_dpp_1})-(\ref{ep_dpp_2}) the DPP in the finite horizon case can be written more compactly as
	\begin{equation}\label{ep_dpp}
	G(t,x) = \underset{\pi \in \A}{\sup} \; \E_{t,x}\left[G(\tau,X_\tau)\right],
	\end{equation}	
	for any stopping time $\tau \in \TT_{t,T}$. 
\end{remark}

\noindent This result helps give context and meaning to the Hamilton-Jacobi-Bellman (HJB) equation, which is the infinitesimal version of the Dynamic Programming Principle (see Theorem \ref{tdpp}). Let us now derive formally the HJB equation, which is in the spirit of Section 3.4.1 in \cite{textbook_pham}. 

\subsection{Formal derivation of the HJB equation}\label{sub_hjb_derivation}

Consider the time $\tau=t+h$ and a constant control $\pi_u = p$, for some arbitrary $p$ in $\A$. Then, from (\ref{ep_dpp}) we have
\begin{equation}\label{ep_optimisation_arb}
G(t,x) \ge \E_{t,x}\left[G(t+h,X_{t+h})\right].
\end{equation}

\noindent By making the assumption that $G\in C^2$, we may apply It\^o's lemma between $t$ and $t+h$ 
$$ G({t+h},X_{t+h})=G(t,x) + \int_{t}^{t+h}\left(\frac{dG}{dt}+\LL^pG\right)(u,X_u^{t,x})du + \text{(local) martingale,}$$

\noindent where $\LL^p$ is the operator associated to the diffusion (\ref{ep_diffusion}) for constant control $p$ given by
\begin{equation}\label{ep_operator}
 \LL^pG = b(t,p,x).D_xG + \frac{1}{2}\tr(\sigma(t,p,x)\sigma'(t,p,x)D_x^2G),
\end{equation}
where $D_x$ and $D_x^2$ are the gradient and Hessian operators respectively. Substituting (\ref{ep_operator}) into (\ref{ep_optimisation_arb}), we obtain
$$0 \ge \E \left[\int_{t}^{t+h}\left(\frac{dG}{dt}+\LL^pG\right)(u,X_u^{t,x})du\right].$$
Dividing by $h$ and sending $h$ to $0$, yields
$$0 \ge \frac{dG}{dt}(t,x)+\LL^pG(t,x),$$
by the Mean-Value Theorem. As this holds for any $p \in A$, we obtain the following inequality
\begin{equation}\label{ep_inequality1}
-\frac{dG}{dt}(t,x)-\underset{p\in \A}{\sup}[\LL^pG(t,x)] \ge 0.
\end{equation}
Suppose that $\pi^*$ is an optimal control, and let $X^*$ denote the process with respect to the control $\pi^*$. Then, from (\ref{ep_dpp}) we have
\begin{equation}
G(t,x) = \E_{t,x}\left[G(t+h,X^*_{t+h})\right].
\end{equation}
By similar arguments as above, we then get
$$-\frac{dG}{dt}(t,x)-\LL^{\pi^*_t}G(t,x)=0,$$
which combined with (\ref{ep_inequality1}) suggests that $G$ should satisfy
\begin{equation}\label{ep_hjb}
-\frac{dG}{dt}(t,x)-\underset{p\in A}{\sup}[\LL^{p}G(t,x)]=0, \quad \fa (t,x) \in [0,T) \times \R^k.
\end{equation}
To relate this back to the optimisation problem (\ref{eb_optimsation}), we also impose the terminal condition associated to this PDE as 
\begin{equation}\label{ep_terminal}
G(T,x)=\frac{x_1^\gamma}{\gamma}, \quad \fa x\in \R^k
\end{equation}
where $x=(x_1, \dots x_k)$.

\subsection{Solution to the finite dimensional optimisation problem}

Using the results from the above section, let's now derive the corresponding HJB equation for the optimisation problem (\ref{eb_optimsation}). Let $(X)_{t\in [0,T]}$ in (\ref{ep_diffusion}), be defined by $X_t\de (\Pi^\pi_t, \, \widetilde{Y}_t^{x_1^n},\dots, \widetilde{Y}_t^{x_n^n},\, Z_t)$, where $\Pi^\pi_t$ is given by (\ref{eb_wealth_process_approx}), $\widetilde{Y}_t^{x_i^n}$ is given by (\ref{eb_y_sde}) and $Z_t$ is given by (\ref{eb_zdynamics}). Moreover, as the parameters $(\kappa, \phi, \sigma,\theta)$ are constant, and $r$ is deterministic, $X$ has a strong unique solution, thus satisfying the conditions of Theorem \ref{tdpp}. Then, the function $G \in C^2$ in (\ref{ep_hjb}) with respect to (\ref{eb_optimsation}) is given by $G(t,w,\widetilde{y}_1, \dots \widetilde{y}_n,z)$, with $t\in [0,T], w>0, \widetilde{y}_1>0 \dots \widetilde{y}_n>0, z>0$ and terminal condition (\ref{ep_terminal}) given by $G(T,w,\widetilde{y}_1, \dots \widetilde{y}_n,z) = \frac{w^\gamma}{\gamma}$. In order to ease notation set $\beta \de a \left(\nu_0 + z \frac{t^{-\alpha-1}}{\Gamma(-\alpha)}+\sum_{i=1}^{n}\widetilde{y}_i\widetilde{q}_i\right)$ and $h_i(t)\de\frac{1}{x_i}(1-e^{-tx_i})$. Thus, the HJB equation (\ref{ep_hjb}) reads
\begin{align}\label{eb_hjb_G}
0 = \underset{p \in \R}{\sup}\; &\bigg \lbrace G_t + G_ww(r_t + p\theta\beta) + \sum_{i=1}^{n}G_{\widetilde{y}_i}\left(h_i(t)\kappa(\phi-z)-x_i\widetilde{y}_i\right) + G_z\kappa(\phi-z) \nonumber \\
&+ \frac{1}{2}G_{ww}w^2p^2\beta + \frac{1}{2}G_{zz}\sigma^2z + \frac{1}{2}\sigma^2z\sum_{i=1}^{n}\sum_{j=1}^{n}G_{\widetilde{y}_i\widetilde{y}_j}h_i(t)h_j(t) \nonumber \\
&+ \sigma^2z\sum_{i=1}^{n}G_{\widetilde{y}_iz}h_i(t) + G_{wz}wp\sigma \rho \sqrt{z\beta} +\sum_{i=1}^{n}G_{w\widetilde{y}_i}\rho w p \sigma \sqrt{z\beta}h_i(t) \bigg \rbrace .
\end{align}
Now, recognising the results from Section \ref{soptimisationsolution}, namely (\ref{edistortion}), B\"auerle and Desmettre use the Ansatz 
\begin{equation}\label{eb_ansatz}
G(t,w,\widetilde{y}_1, \dots \widetilde{y}_n,z)=\frac{w^\gamma}{\gamma}f(t,\widetilde{y}_1, \dots \widetilde{y}_n,z),
\end{equation}
with $f(T,\widetilde{y}_1, \dots \widetilde{y}_n,z)=1, \fa \widetilde{y}_i\in \R, z\in \R_+$. By then substituting (\ref{eb_ansatz}) into (\ref{eb_hjb_G}) the following HJB equation is obtained
\begin{align}\label{eb_hjb}
0 = \underset{p \in \R}{\sup}\; &\bigg \lbrace f_t + f(r_t + p\theta\beta)\gamma + \sum_{i=1}^{n}f_{\widetilde{y}_i}\left(h_i(t)\kappa(\phi-z)-x_i\widetilde{y}_i\right) + f_z\kappa(\phi-z) \nonumber\\
&+ \frac{1}{2}(\gamma-1)\gamma p^2\beta f + \frac{1}{2}f_{zz}\sigma^2z \nonumber + \frac{1}{2}\sigma^2z\sum_{i=1}^{n}\sum_{j=1}^{n}f_{\widetilde{y}_i\widetilde{y}_j}h_i(t)h_j(t) \nonumber \\
&+ \sigma^2z\sum_{i=1}^{n}f_{\widetilde{y}_iz}h_i(t) + f_{z}\gamma p\sigma \rho \sqrt{z\beta} +\sum_{i=1}^{n}f_{\widetilde{y}_i}\gamma \rho w p \sigma \sqrt{z\beta}h_i(t) \bigg \rbrace .
\end{align}
Maximising (\ref{eb_hjb}) in $p$, \bd obtain the following optimal trading strategy $\pi^*$ given by

\begin{equation}\label{eb_opt_rho}
\pi_t^* = \frac{\theta}{1-\gamma} + \frac{\sigma \rho}{1-\gamma}\sqrt{\frac{z}{\beta}}\left(\frac{f_z+\sum_{i=1}^{n}f_{\widetilde{y}_i}h_i(t)}{f}\right)
\end{equation}

\noindent Substituting the maximum point (\ref{eb_opt_rho}) into (\ref{eb_hjb}), yields 
\begin{align}\label{eb_hjb_0}
0 &= f_t + f\left(r_t + \frac{\theta^2\beta}{2(1-\gamma)}\right)\gamma + \sum_{i=1}^{n}f_{\widetilde{y}_i}\left(h_i(t)\kappa(\phi-z)-x_i\widetilde{y}_i\right) + f_z\kappa(\phi-z) \nonumber\\
&+ \frac{1}{2}f_{zz}\sigma^2z + \frac{1}{2}\sigma^2z\sum_{i=1}^{n}\sum_{j=1}^{n}f_{\widetilde{y}_i\widetilde{y}_j}h_i(t)h_j(t)+ \sigma^2z\sum_{i=1}^{n}f_{\widetilde{y}_iz}h_i(t) \nonumber \\
&+\frac{z\gamma \sigma^2 \rho^2}{2(1-\gamma)f}\left(f_z + \sum_{i=1}^{n}f_{\widetilde{y}_i}h_i(t)\right)^2 + \frac{\sigma \rho\theta \gamma}{1-\gamma}\sqrt{z\beta}\left(f_z + \sum_{i=1}^{n}f_{\widetilde{y}_i}h_i(t)\right) \bigg \rbrace .
\end{align}

\noindent Unfortunately (\ref{eb_hjb_0}) is a rather involved partial differential equation (PDE) and has to be solved numerically. However, we can reduce the complexity by utilising the distortion power (\ref{edistortionpower}), which results in cancelling $f^2_z$, $f_{\widetilde{y}_1}f_z$ and $f_{\widetilde{y}_i}^2$ terms by defining the function 
\begin{equation}\label{eb_distortiontransform}
g(t, \widetilde{y_1},\dots, \widetilde{y}_n,z)=f(t, \widetilde{y_1},\dots, \widetilde{y}_n,z)^{1/\delta},
\end{equation}
where $\delta$ is given by (\ref{edistortionpower}). Therefore, substituting (\ref{eb_distortiontransform}) into (\ref{eb_hjb_0}), and dividing by $g^{\delta-1}$, we get

\begin{align}\label{eb_hjb_simplified}
0 &= \delta g_t + g\left(r_t + \frac{\theta^2\beta}{2(1-\gamma)}\right)\gamma +  \delta\sum_{i=1}^{n}g_{\widetilde{y}_i}\left(h_i(t)\kappa(\phi-z)-x_i\widetilde{y}_i\right) +  g_z\delta\kappa(\phi-z) \nonumber\\
&+ \frac{1}{2}g_{zz}\delta\sigma^2z + \frac{1}{2}\delta\sigma^2z\sum_{i=1}^{n}\sum_{j=1}^{n}g_{\widetilde{y}_i\widetilde{y}_j}h_i(t)h_j(t)+ \delta\sigma^2z\sum_{i=1}^{n}g_{\widetilde{y}_iz}h_i(t) \nonumber \\
&+ \frac{\delta\sigma \rho\theta \gamma}{1-\gamma}\sqrt{z\beta}\left(g_z + \sum_{i=1}^{n}g_{\widetilde{y}_i}h_i(t)\right) \bigg \rbrace .
\end{align}

\noindent Let us now simplify $(\ref{eb_hjb_0})$ by taking the case $\rho=0 \implies \delta=1$. From this case we can see that it admits a solution which is given by the Feynman-Kac representation (see \cite{textbook_pham}). Recognising this, \bd obtain the following representation for the function $f$ given by the following theorem.

\begin{theorem}\cite[Theorem 5.4]{bauerle_desmettre_2019_portfolio_optimization_des_baulere} Suppose $\rho=0$, and a classical solution $f$ of the partial differential equation (\ref{eb_hjb_0}) with boundary condition $f(T,\widetilde{y}_1, \dots \widetilde{y}_n,z)=1$ exists for $t\in [0,T], w>0, \widetilde{y}_i>0,z>0$. Then it can be written as
	\begin{equation}
	f(t,\widetilde{y}_1,\dots, \widetilde{y}_n,z)=\E_{t, \widetilde{y}_1, \dots \widetilde{y}_n,z}\left[\exp \left\lbrace\int_{0}^{T}\left(\gamma r_t + \frac{\gamma \theta^2}{2(1-\gamma)}a(\nu_t)\right)dt \right\rbrace  \right]
	\end{equation}	
\end{theorem}

\noindent From this result \bd then obtain the following solution for the optimisation problem (\ref{eoptimisation}) by taking the limit $n\to \infty$ for the case $\rho=0$.

\begin{theorem}\cite[Theorem 5.4]{bauerle_desmettre_2019_portfolio_optimization_des_baulere} Suppose a solution $f$ of the partial differential equation (\ref{eb_hjb_0}) in the case $\rho=0$ with boundary condition $f(T,\widetilde{y}_1, \dots \widetilde{y}_n,z)=1$ exists. The portfolio optimisation problem (\ref{eoptimisation}) with respect to the volatility process (\ref{eb_rough_volatility}), with $\alpha \in (-1,-1/2)$, has an optimal investment strategy given by $\pi_t^*=\frac{\theta}{1-\gamma}$, with value function
\begin{equation}\label{eb_optimisation_solution}
J_0^{\pi^*}(w_0,\nu_0, z_0)=\frac{w_0^\gamma}{\gamma}\underset{n\to \infty}{\lim}\E_{0, \widetilde{y}_1, \dots \widetilde{y}_n,z_0}\left[\exp \left\lbrace\int_{0}^{T}\left(\gamma r_t + \frac{\gamma \theta^2}{2(1-\gamma)}a(\nu_t)\right)dt \right\rbrace  \right]
\end{equation}	
\end{theorem}

Therefore, in the limiting case for $\rho=0\implies \delta=1$, we can indeed see that the finite dimensional approximation approach of the optimisation problem (\ref{eoptimisation}), converges to a solution of the same form as (\ref{edistortion}) and (\ref{ehan_j0}), as expected. 

\chapter{Conclusion}\label{ccl}

In this thesis we examined the Merton portfolio optimisation problem under two different rough Heston models, where the construction of the volatility process played a crucial role in the solution. Moreover, to accommodate the use of stochastic Volterra equations in modelling the financial market dynamics, we permitted the use of a general filtration in the financial market model. Interestingly, by considering such a model, we were able to successfully solve the optimisation problem through the use of the martingale optimality principle. We also showed that historical works developed under the classical framework played an important role in helping to define the Ansatz for the martingale optimality principle. The striking similarities between Markovian and non-Markovian optimisation problems opens up future research considerations for optimisation problems under more general utility functions. Unfortunately, solving the optimisation problem through the use of the martingale optimality principle also has a number of drawbacks, as increasing model complexity under this approach is not easily achieved, which is evident from the dependency on parameter definitions and assumptions in obtaining the solution, consequently impacting model flexibility. To accommodate additional model complexities, such as drift uncertainty and multi-factor models, we also considered an alternative approach which involved a finite dimensional approximation of the rough volatility process. By so doing, the optimisation problem was able to be cast into the classical optimisation framework. This approach opens up opportunities to consider advanced model dynamics coupled with multi-factor rough volatility processes with much more ease, with the only significant drawback being potential computational inefficiencies. With this being said, the approaches discussed in this thesis, combined with the historical works on the distortion transformation, provide a strong foundation to build models capable of handling increasing complexity demanded by the ever growing financial market.



\clearpage
\addcontentsline{toc}{chapter}{References}
\bibliographystyle{unswthesis}

\end{document}